\documentclass{svjour3} 
\smartqed 
\usepackage{indentfirst}
\usepackage[utf8]{inputenc}
\usepackage{amsmath}
\usepackage{amsfonts}
\usepackage{amssymb}
\usepackage{listings}
\usepackage{mathtools}
\usepackage{mathrsfs} 
\usepackage{times}
\usepackage{float}
\usepackage{color}
\usepackage{setspace}
\usepackage{color,soul}

\usepackage{amsmath,amsfonts,amsthm,eucal}
\usepackage{enumerate}
\usepackage[letterpaper,margin=3cm]{geometry}
\usepackage{float}
\usepackage[title]{appendix}
\usepackage{color}
\usepackage{graphics}
\usepackage{stringenc}
\usepackage[
pdfstartview=XYZ,
bookmarks=true,
colorlinks=true,
linkcolor=blue,
urlcolor=blue,
citecolor=blue,
pdftex,
bookmarks=true,
linktocpage=true,   
hyperindex=true
]{hyperref}
\usepackage{lipsum}
\usepackage[FIGBOTCAP,TABTOPCAP,bf,tight]{subfigure}
\usepackage{natbib}
\usepackage[pdftex]{graphicx}
\usepackage{epstopdf}
\usepackage{booktabs} 
\usepackage{lineno}
\usepackage{empheq}

\usepackage{bm,upgreek}
\usepackage{tikz,mathpazo}
\usetikzlibrary{shapes.geometric, arrows}

\usepackage{caption}

\newcommand{\tensor}[1]{\ensuremath{\boldsymbol{#1}}}
\newcommand{\alert}[1]{\textcolor{red}{#1}}

\DeclareMathOperator{\Grad}{\nabla^{\tensor X}}

\DeclareMathOperator{\tr}{tr}

\theoremstyle{remark}

\renewcommand{\vec}[1]{\ensuremath{\boldsymbol{#1}}}

\DeclareMathOperator*{\argmin}{arg\,min}
\newcommand{\parder}[2]{\frac{\partial #1}{\partial #2}}

\usepackage{algorithm}
\usepackage[noend]{algpseudocode}
\usepackage{multirow}

\usepackage{CJKutf8}
\usepackage{kotex}

\usepackage{regexpatch}
\usepackage{listings}

\usepackage{xcolor}

\DeclareFixedFont{\ttb}{T1}{txtt}{bx}{n}{9} 
\DeclareFixedFont{\ttm}{T1}{txtt}{m}{n}{9}  

\usepackage{color}
\definecolor{deepblue}{rgb}{0,0,0.5}
\definecolor{deepred}{rgb}{0.6,0,0}
\definecolor{deepgreen}{rgb}{0,0.5,0}
\definecolor{terminalblack}{rgb}{0.25,0.25,0.25}
\definecolor{f77green}{rgb}{0,0.7,0}
\definecolor{f77blue}{rgb}{0.0,0,0.7}

\newcounter{python}

\newcounter{fortran}

\makeatletter

\lst@UserCommand\lstlistofpython{\bgroup
    \let\contentsname\lstlistpythonname
    \let\lst@temp\@starttoc \def\@starttoc##1{\lst@temp{lop}}%
    \tableofcontents \egroup}
\lstnewenvironment{python}[1][]{%
  \let\c@lstlisting=\c@python
  
  \xpatchcmd*{\lst@MakeCaption}{lol}{lop}{}{}%
  \lstset{language=Python,
  basicstyle=\small,
  otherkeywords={self},
  tabsize=1,
  keywordstyle=\ttb\color{deepblue},
  emph={__init__, update_port},
  emphstyle=\ttb\color{deepred},
  stringstyle=\ttb\color{deepgreen},
  commentstyle=\color{f77green},
  frame=single,
  showstringspaces=false,
  float=htpb,
  captionpos=b,
  numbersep=5pt,
  linewidth=0.99\linewidth,
  xleftmargin=0.01\linewidth,
  breaklines=true
  breakwhitespace=false,
  #1}}
  {}

\lst@UserCommand\lstlistoffortran{\bgroup
    \let\contentsname\lstlistfortranname
    \let\lst@temp\@starttoc \def\@starttoc##1{\lst@temp{lof}}%
    \tableofcontents \egroup}
\lstnewenvironment{fortran}[1][]{%
  \let\c@lstlisting=\c@fortran
  
  \xpatchcmd*{\lst@MakeCaption}{lol}{lof}{}{}%
  \lstset{language=[90]fortran,
  basicstyle=\ttfamily\small,
  showstringspaces=false,
  keywordstyle=\color{f77blue},
  stringstyle=\color{purple},
  commentstyle=\color{f77green},
  breakatwhitespace=false,
  frame=single,
  linewidth=0.99\linewidth,
  xleftmargin=0.01\linewidth,
  breaklines=true
  breakwhitespace=false,
  showstringspaces=false,
  captionpos=b,
  #1}}
  {}

\usepackage{algorithm}
\usepackage{algpseudocode}

\renewcommand{\alert}[1]{#1}

\algdef{SE}[SUBALG]{Indent}{EndIndent}{}{\algorithmicend\ }%
\algtext*{Indent}
\algtext*{EndIndent}

\title{Physics-constrained symbolic model discovery for polyconvex incompressible hyperelastic materials
}

\begin{document}


\author{Bahador Bahmani      \and 
        WaiChing Sun 
}

\institute{Corresponding author: WaiChing Sun \at
Associate Professor, Department of Civil Engineering and Engineering Mechanics, 
 Columbia University , 
 614 SW Mudd, Mail Code: 4709, 
 New York, NY 10027
  Tel.: 212-854-3143, 
  Fax: 212-854-6267, 
  \email{wsun@columbia.edu}        
}

\date{Received: \today / Accepted: date}

\maketitle

\begin{abstract}
We present a machine learning framework capable of consistently inferring mathematical expressions of hyperelastic energy 
functionals for incompressible materials from sparse experimental data and physical laws. 
To achieve this goal, we propose a polyconvex neural additive model (PNAM) that enables us to express the hyperelastic model 
in a learnable feature space while enforcing polyconvexity. An upshot of this feature space obtained via the PNAM is that (1) it is spanned by a set of univariate basis functions that can be re-parametrized with a more complex mathematical form,  and (2) the resultant elasticity model is guaranteed to fulfill the polyconvexity, which ensures that the acoustic tensor remains elliptic for any deformation. 
To further improve the interpretability,  we use genetic programming to convert each univariate basis into a compact mathematical expression. The resultant multi-variable mathematical models obtained from this proposed framework are not only more interpretable 
but are also proven to fulfill physical laws.  By controlling the compactness of the learned symbolic form,  the machine learning-generated mathematical model also requires fewer arithmetic operations than its deep neural network counterparts during deployment. 
This latter attribute is crucial for scaling large-scale simulations where the constitutive responses of every integration point must be updated within each incremental time step. 
We compare our proposed model discovery framework against other state-of-the-art alternatives to assess the robustness and efficiency of the training algorithms and examine the trade-off between interpretability,  accuracy,  and precision of the learned 
symbolic hyperelastic models obtained from different approaches. 
Our numerical results suggest that our approach extrapolates well outside the training data regime due to the precise incorporation of physics-based knowledge.

%
%
%
%

\end{abstract}


\section{Introduction}
Engineering analysis and computer simulations of solid behaviors often involve solvers that predict admissible solutions 
that fulfill a set of constraints. As explained in \citet{kirchdoerfer2016data}, these constraints can be (1) balance principles, 
which are often regarded as ground truth, and (2) constitutive laws, of which the legitimacy of employing such supplemental constraints depends on the specific situation upon which the model is used \citep{boyce2000constitutive}. 
For instance, while vehicle crash simulations for safety analysis \citep{le2018parametric, barbani2014development} 
and those for computer animation and gaming \citep{grinspun2003discrete} can both benefit from improving high fidelity, the demands for robustness, accuracy, efficiency, and risk tolerance in predictions are significantly different for these two applications. 
This difference often leads to different modeling choices, ranging from simplified models that offer speed and robustness, often at the expense of accuracy, to highly sophisticated models or multiscale constitutive updates from representative elementary volumes that exhibit higher fidelity, precision, and accuracy,  even at the expense of 
efficiency \citep{dafalias1984modelling}. 

For soft materials that remain elastic while undergoing large isochoric deformation, modeling frameworks, such as hypoelasticity (cf. \citet{truesdell1955hypo, green1956hypo, freed2010hypoelastic}) and hyperelasticity (cf. \cite{ogden1997non, holzapfel2000new, mihai2017family, mooney1940theory, rivlin1951large}) are some of the more popular choices for these different applications. In the former framework,  one may directly establish a relationship between a pair of stress and strain measures. 
For instance, one may compose an isotropic function that maps the left Cauchy-Green tensor $\tensor{b}$ to the Cauchy stress for isotropic elastic materials \citep{truesdell1955hypo}. Hyperelastic models, on the other hand, provide an alternative strategy where one considers that the deformation increases the Helmholtz free energy, and the stress and its corresponding tangent stiffness are the Jacobian and the Hessian of this free energy with respect to the deformation measure. 
This latter approach provides a convenient and flexible way to fulfill thermodynamics consistency, enforce different symmetries, and, if desired, ensure the stability of the material models by enforcing properties (e.g., growth condition, convexity, polyconvexity) on the elastic energy functionals.  
As such, prior knowledge of the material behaviors, such as the lattice structure of a crystal (for symmetries) and the shape memory effect of certain alloys (for multiple potential wells), can be easily incorporated into the modeling process. 
If this prior knowledge can be incorporated by deducing the specific \textit{form} of the energy functional, then the last remaining task to complete the model is to find the actual \textit{parametrization} of the model that fulfills all constraints of both experimental data and prior physical knowledge.
%
%
%

In addition to hand-crafted mathematical expressions, a variety of alternatives, such as artificial neural networks (ANNs) \citep{vlassis2020geometric, liu2020generic, thakolkaran2022nn, klein2022polyconvex, tac2022data}, Gaussian processes \citep{frankel2020tensor,fuhg2022physics}, and symbolic regression (SR) \citep{abdusalamov2023automatic}, are often used to generate closure of the constitutive laws. 
In the former case,  feedforward ANNs may provide the expressivity (the ability to fit any complex data) necessary to yield precise models \citep{hornik1989multilayer, hsu2021approximation}. However, underfitting,  overfitting,  the lack of interpretability,  and the incompatibility with  physics constraints could all negatively impact the quality of the learned models. 
While the underfitting and overfitting issues can be circumvented with hyperparameter tuning, both the lack of interpretability and the incompatibility with known physics constraints are issues that make the resultant models not feasible for the intended engineering applications, especially of high consequences,  such as patient-specific simulations or design of structural components for civil infrastructure systems. 
For example, \citet{shen2004neural} and \citet{liang2008neural} proposed training a hyperelastic energy functional for elastomeric foams. \cite{liang2008neural} use strain invariants as inputs for an ANN and conducted training based on the calculated stress as labeled data for supervised learning. 
%
Meanwhile, \citet{le2015computational} introduce neural network hyperelastic models to upscale constitutive responses for representative elementary volumes where the energy, stress, and stiffness  are obtained from ANNs.
There has been a rapidly growing body of work on hyperelastic models parametrized by ANNs trained using Sobolev norms \citep{vlassis2020geometric, vlassis2021sobolev, vlassis2022molecular}, neural ordinary differential equations \citep{tac2022data} as well as multi-objective optimization problems that attempt to fulfill all data and physics constraints \citep{klein2022polyconvex, fernandez2022material}.

This paper aims to formulate a feasible machine learning framework that can consistently generate hyperelastic models 
whose properties can be easily interpreted and fulfill all known physics constraints.  A critical technical barrier we would like to overcome 
is the difficulty of determining the expression tree in the SR that simultaneously satisfies the polyconvexity of the learned model. 
While there have been works on generating polyconvex neural network models via neural networks,  to the authors' best knowledge,  
this contribution is the first attempt to develop a machine learning algorithm to generate a polyconvex hyperelastic model expressed via mathematical expressions. 
To achieve this objective,  we use a parameterization strategy similar to the neural additive model \citep{agarwal2021neural} and the recently proposed quadratic extension \citep{bahmani2023discovering},  of which the feature space is spanned by univariate functions obtained from neural networks while improving the machine learning algorithm by introducing physics constraints and polyconvexity to ensure desirable properties of the learned hyperelastic models. 
To improve interpretability without comprising expressivity and accuracy,  we introduce an additional step where we use SR on the feature space such that it can be approximated by basis functions expressed analytically. 
As polyconvexity can be guaranteed by energy functionals written in the specific additive form (cf. \citet{hartmann2003polyconvexity}), this feature enables us to express the resultant polyconvex hyperelastic model as a function of the features of the strain invariants.

\subsection{Reviews of Physics Constraints for Elasticity Models}
Previous machine learning models have been trained 
with prior physical knowledge incorporated as constraints (e.g., \citep{teichert2019machine,liu2020generic,masi2021thermodynamics,tac2022data,vlassis2022molecular}).
The basic strategy is similar to those used in physics-informed neural network paradigms considered seminal by many for 
solving partial differential equations where a variety of loss functions are employed to ensure that 
the learned solution satisfies physics constraints \citep{raissi2019physics,lagaris1998artificial}.
However, as the training is designed to minimize,  not eliminate,  discrepancies,  there is a possibility of violating these constraints when dealing with data not used during training,  especially in the extrapolation regime \citep{bronstein2021geometric}. 
Furthermore,  enforcing multiple constraints into the loss function may lead to a multi-objective optimization problem
where gradient conflicts among different objectives may further complicate the search for global optima \citep{yu2020gradient, bahmani2021training},  which is an NP-hard problem \citep{jin2016provable}.

Several neural network architectures have been proposed to explicitly incorporate all or a subset of physical laws by design for constitutive modeling \citep{heider2020so, xu2021learning,linka2021constitutive,tac2022data,as2022mechanics,chen2022polyconvex, cai2023equivariant}.  The by-design strategy can be selecting the optimal parameterization of input variables (e.g., using strain invariants instead of the strain tensor for isotropic materials) or modifying neural network architectures to preserve symmetry,  invariance,  and equivariance. 
By fulfilling the proposed physics constraints by design,  the learned models that inherently fulfill the physics constraints are more robust, especially in the data-limited regime.  
In fact,  this incorporation of physics constraints is in line with the history of hand-crafted hyperelastic models in which material symmetry has already been heavily leveraged to yield a specific form of mathematical expressions that reduces the number of independent variables, enforce symmetry and thermodynamics constraints \citep{ogden1997non, schroder2003invariant, holzapfel2000new},  
and induce desirable properties, such as polyconvexity \citep{hartmann2003polyconvexity,
schroder2003invariant,
schroder2005variational,
schroder2010poly} and quasiconvexity \citep{ball1987does, shirani2022convexity}. 
In this paper, we will adopt this latter by-design strategy while leveraging the power of the neural additive model 
and symbolic regression to further improve the model obtained from the machine learning algorithm.

\subsection{Reviews on Interpretable Machine Learning Constitutive Laws}
Learnable parameters, such as weights and biases, parametrize the learned function obtained from training neural networks. As such, the current trend of increasingly deep and large neural networks often leads to significant challenges in interpreting and 
examining the \textit{global} property of machine learning models \citep{dayhoff2001artificial, oh2019towards}. For instance, while it is possible to use a sampling technique to test the robustness of the learned model for a set of strain inputs against constraints, satisfying the physics constraints for a subset of data points only estimates the population loss. Model accuracy shown in the sampling test is only a 
necessary but not sufficient condition for generalizability. 

An obvious strategy to circumvent this issue is to derive alternative parameterization that may lead to more compact mathematical expressions where analyses (such as calculating the acoustic wave speed and detecting the loss of ellipticity) typically performed on hand-crafted models can be conducted. In contrast to ANN-based methods, SR methods are free-form approaches where the equation form is discovered in a data-driven manner using gradient-free methods like genetic programming. The application of SR algorithms in data-driven mechanics has proven effective in discovering yield functions for plasticity \citep{versino2017data,bomarito2021development,park2021multiscale}. While initial attempts do not enforce physics knowledge, data augmentation via physical intuitions has shown improvements in the SR performance \citep{versino2017data}. In a recent work by \cite{abdusalamov2023automatic}, a SR method for discovering hyperelastic materials is introduced, directly utilizing energy functionals. This approach offers advantages in terms of thermodynamics consistency. However, it can be computationally expensive due to the requirement for symbolic gradient calculations during optimization iterations.

The advantage of SR machine learning methods is their ability to provide explicit equation forms, which are often simpler than neural network operations. However, various challenges have hindered their popularity compared to ANN-based methods in mechanics. 
Firstly, their lack of scalability, especially in multidimensional data settings, is attributed to the combinatorial nature of their search space. Secondly, incorporating mechanistic constraints, especially those related to gradient operations like ellipticity, is not a straightforward task in these algorithms. This difficulty arises from their use of gradient-free optimizers and the costly process of symbolic gradient calculations, in contrast to the efficiency of automatic differentiation methods used in ANN-based methods.

Sparse regression within a predefined library of modes \citep{brunton2016discovering} offers a potential method to bridge the gap between scalability and interpretability in model identification. \citet{flaschel2021unsupervised} and \citet{wang2021inference} develop a material discovery formulation from a predefined library of material models. Their approach indirectly discovers a material model from displacement and force data over the boundary of the material sample. Similarly, \cite{linka2023new} employ a library of modes inspired by classical constitutive models and prior physical knowledge to directly learn energy functionals from strain and stress data. Each mode's contribution to the final prediction is trained using a gradient-based optimizer with automatic differentiation as the backbone algorithm. Nevertheless, such predefined modes may introduce significant bias in the modeling and might restrict the learning of complex modes not present in the library.   For instance, \citet{wang2022establish} parameterize the stress tensor as a polynomial function of the strain tensor. However, since physics constraints, such as material symmetry and thermodynamics constraints, are not explicitly enforced in the formulation, the model is not guaranteed to be compatible with these constraints.

\remark{\textbf{Alternative approaches for elasticity problems}
The model-free or distance minimization method \alert{\citep{kirchdoerfer2016data}} is extended to finite deformation \citep{nguyen2018data,platzer2021finite}, eliminating the need for any model assumptions. Another model-free approach known as What-You-Prescribe-Is-What-You-Get (WYPIWYG) \citep{crespo2017wypiwyg} is also employed. Following the WYPIWYG idea, methods based on spline shape functions are also introduced \citep{amores2019average,moreno2020reverse,akbari2022reverse}. Some studies formulate the learning of constitutive laws as a manifold learning problem, searching for the response surface rather than using conventional surrogate models for input-to-output mapping \citep{ibanez2018manifold,he2021deep,gonzalez2020data,bahmani2022manifold,bahmani2023distance}.}

\subsection{Notations and Organization of the Remaining Paper}
The remaining content of this paper is organized as follows. In Section \ref{sec:mechEq}, we review crucial elements needed to construct our model structure, adhering to physical knowledge and mechanical properties of incompressible hyperelastic materials. These properties include isotropy, material objectivity, polyconvexity, and coercivity. Our model discovery method is outlined in Section \ref{sec:modelDisc}, where we summarize our two-step approach. We then describe the specific structure of the neural network and SR algorithms used to ensure polyconvexity of the final discovered energy functional. In Section \ref{sec:expReduce}, we provide reduced forms of the proposed formulation to handle common experimental setups, which will be useful for calibrating the model from experimental data. To demonstrate the effectiveness of our framework, we find two symbolic models for real and synthetic data in Section \ref{sec:numExamp}. Additionally, we provide a discussion on the formal analysis of the discovered models in Section \ref{sec:discuss}.

As for notations and symbols, bold-faced and blackboard bold-faced letters denote tensors (including vectors which are rank-one tensors); 
the symbol '$\cdot$' denotes a single contraction of adjacent indices of two tensors 
(e.g.,\ $\vec{a} \cdot \vec{b} = a_{i}b_{i}$ or $\tensor{c} \cdot \tensor{d} = c_{ij}d_{jk}$); 
the symbol `:' denotes a double contraction of adjacent indices of tensors of rank two or higher
(e.g.,\ $\mathbb{C} : \vec{\varepsilon}$ = $C_{ijkl} \varepsilon_{kl}$); 
the symbol `$\otimes$' denotes a juxtaposition of two vectors 
(e.g.,\ $\vec{a} \otimes \vec{b} = a_{i}b_{j}$)
or two symmetric second-order tensors 
[e.g.,\ $(\tensor{\alpha} \otimes \tensor{\beta})_{ijkl} = \alpha_{ij}\beta_{kl}$]. 
We also define identity tensors: $\tensor{I} = \delta_{ij}$ and $\mathbb{I} = (\delta_{ik}\delta_{jl} + \delta_{il}\delta_{jk})/2$, where $\delta_{ij}$ is the Kronecker delta.
As for sign conventions, unless specified, tensile stress and dilative pressure are considered positive.

\section{Hyperelasticity Formulation}
\label{sec:mechEq}
In this section, we establish the theoretical foundation upon which we construct our modeling structure, ensuring the incorporation of physical knowledge in the model. First, we delve into the kinematics of finite strain elasticity. Next, we review essential conditions for incorporating physically or empirically inspired constraints, such as polyconvexity. Finally, we derive the most general form of the energy functional, introducing unknown functions that will be parametrized by appropriate hypothesis classes in the following section.

\subsection{Kinematics of Finite Deformation}
\label{sec:kinematics}
For completeness, we briefly review the kinematics of a continuum, which is the input of a path-independent elastic energy functional. 
Recall that the motion of a material point at the reference configuration can be described by the vector field $\vec{\phi}(\vec{X}): \mathbb{R}^3 \to \mathbb{R}^3$ which moves points $\vec{X} = X_I \vec{E}_I$ in the reference configuration $\Omega_0$ (i.e.,  $X \in \Omega_{0}$) to locations $\vec{x}=x_i \vec{e}_i$ in the current configuration $\Omega$. The deformation gradient tensor $\tensor{F}$, the primary measure of deformation, is the tangent operator of the motion $\vec{\phi}$, i.e., 
\begin{equation}
    \tensor{F} = \Grad \vec{\phi}; \quad F_{iJ} = \parder{\vec{\phi}_i}{\vec{X}_J}.
\end{equation}
The stretch vector $\vec{\lambda}^N$ along the direction of the unit vector $\vec{N}$ at $\tensor{X}\in \Omega_0$ is defined as, 
\begin{equation}
    \vec{\lambda}^N = \tensor{F} \vec{N}; \quad \lambda^N_i = F_{iJ} N_J.
\end{equation}
%
For practical reasons, one may prefer to use the right Cauchy-Green tensor $\tensor{C}$ as the deformation measure,
\begin{equation}
    \tensor{C} = \tensor{F}^T\tensor{F}; \quad C_{IJ} = F_{Ii} F_{iJ}.
\end{equation}
The right Cauchy-Green tensor is symmetric and positive-definite which is more favorable for numerical calculations. Moreover, it is fully described with respect to the reference coordinate system which may ease analytical derivations by avoiding conversion between spatial and material coordinate systems.
The first three principal invariants of the right Cauchy-Green tensor are calculated as follows,
\begin{align}
    &I_1 = \tr(C) = C_{II},\\
    &I_2 = \tr(\text{adj}( \tensor{C})) = \frac{1}{2} 
    \left(
        I_1^2 - \tr(\tensor{C}^2)
    \right) = 
    \frac{1}{2} \left(C_{II}C_{JJ} - C_{IJ}C_{IJ}
    \right),\\
    &I_3 = \det(\tensor{C}) = J^2,
\end{align}
where $J$ is the Jacobian of the deformation gradient, i.e., $J = \det(\tensor{F})$, and $\tr(\cdot)$ and $\det(\cdot)$ are trace and determinant operators, respectively. In these relations, adjugate operator is defined as $\text{adj}(\tensor{C}) = \text{cof}(\tensor{C})^T =\det(\tensor{C}) \tensor{C}^{-1}$, where $\text{cof}(\cdot)$ is the cofactor operator. The importance of the tensor representation based on its invariants will be clarified later. For isotropic hyperelastic materials, these three invariants are sufficient to predict the elastic stored energy and the corresponding stress measure due to coaxiality.

\subsection{Physics Constraints for Isotropic Elastic Materials}
Presumably, one may, for instance, develop constitutive theories by establishing relations between the deformation gradient and the first Piola–Kirchhoff stress $\tensor{P}(\vec{F})$. However, caution must be exercised to avoid violating 
physical principles such as thermodynamic consistency \citep{truesdell2004non}. 
Here, our focus is on the modeling of Green-elastic (hyperelastic) materials that postulate the existence of the Helmholtz free energy. 

\subsubsection{Thermodynamic Consistency for Green-elastic Materials}
Assuming that a material produces no entropy locally \citep{truesdell1992first}, then the material is perfectly elastic. 
In this case, the second law of thermodynamics, which requires non-negative internal dissipation, is fulfilled by the existence of the Helmholtz free energy, i.e., 
\begin{equation}
    \mathscr{D} = \tensor{P}:\dot{\tensor{F}} - \dot{W}(\tensor{F}) = \tensor{P}:\dot{\tensor{F}} - \parder{W(\tensor{F})}{\tensor{F}}:\dot{\tensor{F}} \ge 0,
\end{equation}
where $W$ is the Helmholtz free energy, and by definition, since the dissipation
is always zero for perfectly elastic materials,   we have, 
\begin{equation}
    \tensor{P} = \parder{W(\tensor{F})}{\tensor{F}}.
\end{equation}

\subsubsection{Objectivity and Frame Indifference}
The free energy functional $W$ must be invariant with respect to any rigid rotation of the reference coordinate system. This is equivalent to saying,
\begin{equation}
    W(\tensor{F}) = W(\tensor{Q} \tensor{F}) \quad \forall \tensor{Q} \in \text{SO(3)},
\end{equation}
where $\tensor{Q}$ is any arbitrary rotation tensor belonging to the special orthogonal group SO(3), i.e.., $\tensor{Q}^T\tensor{Q} = \tensor{I}$. This requirement can be satisfied if one defines the strain energy functional solely based on the right Cauchy-Green deformation tensor, i.e., $\bar{W}(\tensor{C}) = W(\tensor{F})$, which leads to the following relation,
\begin{equation}
    \exists \bar{W}(\tensor{C}) \ \text{s.t.} \  \tensor{P} = \parder{\bar{W}(\tensor{C})}{\tensor{F}} = 2 \tensor{F} \parder{\bar{W}(\tensor{C})}{\tensor{C}}
    \Longrightarrow \text{frame indifference}.
\end{equation}

\subsubsection{Isotropy Condition}
An isotropic material exhibits the same strain-stress response under a symmetry transformation, i.e., 
\begin{equation}
   W(\tensor{F}) = W( \tensor{F} \tensor{Q}^T) \quad \forall \tensor{Q} \in \text{SO(3)}.
\end{equation}
For the free energy written in terms of $\tensor{C}$, the isotropy of the constitutive responses implies that,
\begin{align}
    \bar{W}(\tensor{C}) = \bar{W}(\tensor{Q}\tensor{C}\tensor{Q}^T) \quad \forall \tensor{Q} \in \text{SO(3)}
    \Longrightarrow \text{isotropic material}.
\end{align}
From the representation theorem for invariants \citep{gurtin1982introduction,holzapfel2002nonlinear} one may show that this constraint is satisfied if the free energy is expressed as a function of only the principal invariants, i.e., 
\begin{align}
    &
    \exists \ \psi(I_1(\tensor{C}), I_2(\tensor{C}), I_3(\tensor{C}))
    \ \text{s.t.} \
    \tensor{S} = 2 \parder{\psi(I_1, I_2, I_3)}{\tensor{C}} \Longrightarrow \text{Isotropic Material},\\
    &\parder{\psi(I_1, I_2, I_3)}{\tensor{C}} = 
    \left(
        \parder{\psi}{I_1} + I_1 \parder{\psi}{I_2}
    \right) \tensor{I}
    - \parder{\psi}{I_2} \tensor{C}
    + I_3 \parder{\psi}{I_3} \tensor{C}^{-1},
    \label{eq:strs-general}
\end{align}
where $\tensor{S}$ is the second Piola–Kirchhoff stress and $\tensor{P} = \textbf{F}\tensor{S}$.

\subsubsection{Incompressibility Condition}
\alert{In this paper,  we limit our focus to deducing the mathematical expression of  elastic stored energy functionals 
for incompressible materials.  A Material is considered incompressible when it only deforms in an isochoric manner,  i.e.,  the 
$\det(\tensor{F}) = J = 1$.  For practical purposes,  the constitutive responses of many solids that exhibit significant isochoric deformation with negligible volumetric deformation (e.g., rubber), as well as liquids in room temperature (e.g., water), are idealized as incompressible \citep{ogden1997non,  boyce2000constitutive,  holzapfel2002nonlinear}.  
In these cases,  one may introduce a scalar variable $p$ (hydrostatic pressure) that serves as a Lagrange multiplier to account for the energy required to maintain this incompressibility constraint. } 
Under this condition,  the second Piola–Kirchhoff stress can be written as, 
\begin{align}
	\label{eq:incomp-energy}
    &\psi = \psi^{\text{uc}}(I_1, I_2) - \frac{1}{2} p U(I_3) \Longrightarrow \text{incompressible \& isotropic material},
    \\ \label{eq:incomp}
    &\tensor{S} = 
    2 \parder{\psi^{\text{uc}}(I_1, I_2)}{\tensor{C}}
    -
    p \parder{U}{\tensor{C}},
    \\
    & \parder{\psi^{\text{uc}}(I_1, I_2)}{\tensor{C}} = 
    \left(
        \parder{\psi^{\text{uc}}}{I_1} + I_1 \parder{\psi^{\text{uc}}}{I_2}
    \right)\tensor{I}
    - \parder{\psi^{\text{uc}}}{I_2} \tensor{C},
    \label{eq:uc-stress-incomp}
\end{align}
where $U(I_3)$ is the energy contribution to penalize the incompressibility constraint,  i.e.,  $|I_3| \to 1$.  

There are various expressions for $U(I_3)$ such as $U(I_3) = I_3 - 1$ and $(J-1)^2 / 2$ to enforce incompressibility \citep{hartmann2003polyconvexity}.  In our case,  we use  $U(I_3) = I_3 - 1$ such that $\partial U/ \partial \tensor{C} = I_3 \tensor{C}^{-1}$ \citep{holzapfel2002nonlinear}.

\subsubsection{Solution Existence and Uniqueness: Polyconvexity Condition}
In this section, we will show that the combination of convex functions of principal invariants is a subclass of polyconvex functions with respect to the deformation gradient, which leads to the existence of solutions for the elasticity boundary value problem.

Let us consider the functional $I(\vec{u})$ defined below,
\begin{equation}
    I(\vec{u}) = 
    \int_{\Omega_0} W\left(\Grad{\vec{u}}\right) d\Omega,
\end{equation}
where $\vec{u}(\vec{X}):\Omega_0\to \mathbb{R}^3$ is the displacement vector field defined over the open set $\Omega_0 \subset \mathbb{R}^3$. The stationary points of this functional satisfy the equilibrium equations of nonlinear elasticity for a homogeneous body under zero body forces. However, an arbitrary free energy functional may not guarantee the existence of minimizers. Convexity of the free energy with respect to the deformation gradient guarantees this existence and uniqueness \citep{hill1957uniqueness}. However, the uniqueness of the solution is too restrictive and not physical in bifurcation scenarios such as buckling. A less restrictive condition is polyconvexity,  which is a sufficient condition for the global existence of the solution \citep{marsden1994mathematical}.

\definition{Convexity:} A function $f(\vec{x}):\mathbb{D} \to \mathbb{R}$ is called convex with respect to its input argument when $f(\lambda \vec{x} + (1- \lambda) \vec{y}) \le \lambda f(\vec{x}) + (1-\lambda)f(\vec{y})$ for any arbitrary $\vec{x}, \vec{y} \in \mathbb{D}$ and $0 < \lambda < 1$.  The convexity condition becomes strict when the equality part becomes inadmissible. Here, $\mathbb{D}$ resembles the function domain, i.e., $\mathbb{D} = \text{dom}(f)$, which is, in this paper, a subset of real-valued vectors or tensors of rank two.

\definition{Rank-one convexity and ellipticity:}
A \alert{twice differentiable} free energy $W(\tensor{F}) = \psi(\tensor{C})$ leads to an elliptic system iff the Legendre-Hadamard or rank-one condition holds \citep{ball1976convexity,ball1980strict,kuhl2006illustration}:
\begin{equation}
    (\vec{M} \otimes \vec{m}): \frac{\partial^2 W}{\partial \tensor{F}\partial \tensor{F}}: (\vec{M} \otimes \vec{m}) \ge 0 \quad \forall \vec{M}, \vec{m} \in \mathbb{R}^3,
\end{equation}
where $\vec{M}$ and $\vec{m}$ are arbitrary vectors in the reference and current coordinate systems, respectively.  The ellipticity (strongly elliptic) condition is satisfied if the inequality is strictly positive. Rank-one convexity is a sufficient condition for material stability and the well-posedness of the elasticity boundary value problem. Notice that the ellipticity requirement does not necessarily result in convexity of the strain energy functional, so multiple minimizers can be found.

\definition{Quasiconvexity:} Morrey's quasiconvexity condition reads,
\begin{equation}
    \int_{D} 
    W\left(
    \tensor{F}_0 + \Grad \vec{v}(\vec{x}) 
    \right)
    d\vec{x} \ge
    |D| W(\tensor{F}_0),
\end{equation}
for all bounded open sets $D \in \mathbb{R}^3$ and all vector fields $\vec{v}(\vec{x})$. The quasiconvexity condition along with certain growth and continuity conditions result in the existence of minimizers \citep{morrey1952quasi}.

\definition{Polyconvexity:} The free energy $\psi(\tensor{F})$ is polyconvex if and only if (iff) it is convex with respect to $(\tensor{F}, \text{adj} \ \tensor{F}, \det \tensor{F})$:
\begin{equation}
    \exists \ \text{convex} \ f \ \text{s.t.} \ \psi(\tensor{F}) = f(\tensor{F}, \text{adj} \ \tensor{F}, \det \tensor{F}), \quad \forall \tensor{F} \in \mathbb{R}^{3\times3}
    \Longleftrightarrow
    \text{polyconvex} \ \psi(\tensor{F}).
\end{equation}
Polyconvexity results in quasiconvexity but not necessarily convexity. 

Notice that all rank-one convexity, quasiconvexity, and polyconvexity conditions are sufficient conditions for the existence of minimizers. However, working with polyconvexity is preferable since, compared to the others, it is more straightforward to be applied as it has locality without any dependence on arbitrary objects (i.e., $\vec{m}$ and $\vec{M}$) other than the deformation gradients. 
\alert{
In summary, the sequence of conditions presented below illustrates a progression where a condition on the right results from a condition on the left, with arrows indicating these implications. For example, quasiconvexity leads to rank-one convexity:
}
\begin{equation}
    \fbox{\text{polyconvexity}} \rightarrow 
    \text{quasiconvexity} \rightarrow
    \text{rank-one convexity}
    \rightarrow
    \text{existence of minimizers.}
\end{equation}

\begin{lemma}\label{lemma:addetivNeff}
A subclass of polyconvex functions can be constructed by convex functions $f_1(\tensor{F})$, $f_2(\text{adj}\tensor{F})$, $f_3(\det \tensor{F})$ in the additive fashion $\psi = f_1(\tensor{F}) + f_2(\text{adj}\tensor{F}) + f_3(\det \tensor{F})$.
\end{lemma}
\begin{proof}
For brevity, we refer readers to \citet{schroder2003invariant} for details.
\end{proof}

\begin{lemma}\label{lemma:composConvx}
If a tensor-valued function $h(\tensor{A}):\mathbb{R}^{3\times3} \to \mathbb{R}$ and a scalar-valued function $g(a):\mathbb{R}\to\mathbb{R}$ are convex \alert{and $g(\cdot)$ is a non-decreasing function, } then $g \circ h$ is convex.
\end{lemma}

\begin{proof}
For two arbitrary tensors $\tensor{A}, \tensor{B} \in \mathbb{R}^{3\times 3}$
\begin{align}
    (g\circ h)(\lambda \tensor{A} + (1-\lambda) \tensor{B}) &=
    g\left(
        h(\lambda \tensor{A} + (1-\lambda) \tensor{B})
    \right)
    \\
    &\le
    g\left(
        \lambda h(\tensor{A}) + (1 - \lambda) h(\tensor{B})
    \right) \quad \text{using convexity of} \ h \text{ and } g \text{ non-decreasing}
    \\
    &\le
    \lambda g(h(\tensor{A})) + (1 - \lambda) g(h(\tensor{B}))  \quad \text{using convexity of} \ g.
\end{align}
\end{proof}

\begin{lemma}\label{lemma:addetivInvars}
If $\psi_1(I_1)$, $\psi_2(I_2)$, and $\psi_3(I_3)$ are convex \alert{and non-decreasing} functions, then $\psi = \psi_1(I_1) + \psi_2(I_2) + \psi_3(I_3)$ is polyconvex in $\tensor{F}$.
\end{lemma}

\begin{proof}
$I_1(\tensor{F})$, $I_2(\text{adj} \tensor{F})$, and $I_3(\det \tensor{F})$ are convex functions \alert{(see Appendix \ref{appx:convx} for details)}. Based on the composition lemma $\psi_1(\tensor{F})$, $\psi_2(\text{adj} \tensor{F})$, and $\psi_3(\det \tensor{F})$ are convex as well. Hence, $\psi$ is polyconvex in $\tensor{F}$.
\end{proof}

In this work, the polyconvexity constraint is guaranteed by choosing parameterized functions for $\psi_1(I_1)$ and $\psi_2(I_2)$ which are convex and non-negative; this choice will be clarified later. We use an additive structure in the unconstrained part of the free energy in Eq.~\eqref{eq:incomp-energy}, i.e., $\psi^{\text{uc}}(I_1, I_2) = \psi_1(I_1) + \psi_2(I_2)$.

\subsubsection{Coercivity (Growth) Conditions}
As mentioned earlier, certain coercivity (growth) conditions along with polyconvexity result in the existence of global minimizers. These conditions ensure that infinite strains result in unbounded stresses. The growth condition is satisfied if,
\begin{equation}
    \psi(\tensor{F}) \ge 
    \alpha \left(
    ||\tensor{F}||^p + ||\text{adj} \tensor{F}||^q + (\det \tensor{F})^r
    \right) + \beta,
\end{equation}
for $\alpha>0$, $\beta>0$, $p \ge 2$, $q \ge \frac{p}{p-1}$,  and $r > 1$. The matrix norm is $||\tensor{F}||^2 = \tr(\tensor{F}^T \tensor{F})$. For incompressible materials, this inequality can be simplified to, 
%
\begin{equation}
    \psi(\tensor{F}) \ge
    \alpha (I_1^{p} + I_2^q) + \alpha + \beta.
    \label{eq:growth}
\end{equation}

In this work, we do not explicitly enforce the coercivity condition, however, we will check the admissibility of the discovered equations according to the coercivity condition.

\remark{
Compressing material towards zero volume should require an infinite amount of energy. Thus one needs to enforce $\psi \to \infty$ when $J \to 0^+$ if the material is compressible.
}

\subsubsection{Stress Free and Positivity conditions}
 The free energy functional should remain non-negative for all possible deformations. 
For practical purposes,  it is often desirable to assume that the initial reference configuration does not exhibit any residual stress 
when  $\tensor{C} = \tensor{F} = \tensor{I}$ (cf. \citet{hoger1986determination}).  These two conditions can be expressed as, 
\begin{align}
    &W(\tensor{F}) \ge 0,\\
    &W(\tensor{F})|_{\tensor{F} = \tensor{I}} =0, \ \tensor{S}(\tensor{F})|_{\tensor{F} = \tensor{I}} = \bold{0}.
\end{align}
The non-negative condition can be satisfied by selecting parameterizations that lead to non-negative free energy functionals.  These parameterizations will be elaborated further in latter sections. 

\subsubsection{Final Form of Hyperelastic Energy Functional for Model Discovery}
\alert{To fulfill the physical constraints and assumptions we listed in the previous sections,  we limit the mathematical expression of the learned model to be in the following form, }
\begin{equation}
    \psi^{\text{final}}(I_1, I_2) = 
    \underbrace{\psi_1(I_1) + \psi_2(I_2)}_{\psi(I_1, I_2)} - \underbrace{\frac{1}{2}p(I_3 - 1)}_{\text{incompressibility}} -
    \underbrace{\alpha_0 (I_3 - 1)}_{\psi_{P_0} \text{zero stress}}
    - \underbrace{(\psi_1(3) + \psi_2(3))}_{\psi_0 \ \text{zero energy}}\label{eq:final}.
\end{equation}
\alert{In particular,  we hypothesize that the 
additive decomposition of the unconstrained energy functional ,  $\psi(I_1, I_2) = \psi_1(I_1) + \psi_2(I_2)$ is valid. 
As discussed in \citet{agarwal2021neural},  restricting the energy functional to take the form of a linear combination of input variables may reduce the expressivity.  This assumption,  however,  can also reduce the difficulty of the symbolic regression by reducing the dimensionality of the learned functions.  This trade-off will be further demonstrated in the numerical examples.  By restricting the learned model to take the form of Equation \eqref{eq:final},  one may enforce the zero-stress condition by calculating the parameter $\alpha_0$ such that the four terms in Equation \eqref{eq:final} canceling out each other at the reference configuration, i.e., }
\begin{align}
    &\alpha_0 = \frac{d \psi_1}{d I_1}|_{I_1 = 3} + 2 \frac{d \psi_2}{d I_2}|_{I_2 = 3}.
\end{align}
One can straightforwardly show that the constructed strain energy satisfies the zero-stress condition at the undeformed state, i.e., when $I_1= I_2 = 3$, $I_3 = 1$, and $p=0$.

Adopting a similar cancellation/counterbalancing strategy used in  \citet{as2022mechanics} and \citet{
chen2022polyconvex},  we formulate the energy functional such that a counterbalance term is added to ensure 
 that the stress-free condition corresponds to the undeformed reference configuration by construction.  This is done as follows. 
Consider an energy functional $\psi(I_1, I_2)$ that does not necessarily meet the stress-free condition in its undeformed state.  According to Equation \eqref{eq:strs-general},  the non-zero stress state is as follows,
\begin{equation}
\left(
\parder{\psi}{I_1}(I_1=3, I_2=3) + 2\parder{\psi}{I_2}(I_1=3, I_2=3) 
\right)
\tensor{I}.
\end{equation}
By introducing an extra energy component to the energy functional, which operates independently of $\psi(I_1, I_2)$, we can generate an equivalent stress value with an opposite sign. This approach effectively neutralizes the non-zero stress initially observed. 
For the incompressible scenario, an effective approach is to formulate the additional energy contribution solely as a function of $I_3$, which remains unaffected by the deformation state (namely, $I_1$ and $I_2$).  A practical example of this is setting $\psi_{{P_0}} = \alpha_0 (I_3 - 1)$, where the parameter $\alpha_0$ is adjustable.   This adjustment ensures that the total energy functional, represented by $\psi(I_1, I_2) - \psi_{{P_0}}$,  equals zero in the undeformed state:
\begin{align}
    &\alpha_0 = \parder{\psi}{I_1}(I_1 = 3, I_2 = 3)
+ 
2 \parder{\psi}{I_2}(I_1 = 3, I_2 = 3).
\label{eq:zero-stress-general}
\end{align}

We will formulate the supervised learning problem such that the resultant expression is compatible with the form in Eq.~\eqref{eq:final}. As we will discuss in next section, this setting limits the expression trees included in the combinatorial optimization and hence reduces the difficulty of the SR problem. 

\section{Interpretable Model Discovery Compatible with Physics Constraints}
\label{sec:modelDisc}
Deriving a physics-constrained model consistent with the proposed form in Eq.~\eqref{eq:final} directly through a single SR algorithm remains challenging for the following reasons. 
Firstly, it is challenging to straightforwardly restrict each expression tree to yield a convex function over the entire data range and beyond. Ensuring convexity can be complex and may lead to sub-optimal results. Secondly, as illustrated in Figure \ref{fig:framework}, the stress objective function requires access to the gradients of the energy functional. This requirement can significantly increase the computational time of SR, especially when dealing with deep expression trees with many terms.
Finally, there is a high likelihood of encountering expression trees with unstable functions, thus further complicating the SR training process. 

To overcome these well-known challenge for SR, we modify our recently developed interpretable data-driven approach for uncovering yield surfaces \citep{bahmani2023discovering} such that the learning process now takes into account physics constraints. Our divide (step 1 in Figure \ref{fig:framework}) and conquer (step 2 in Figure \ref{fig:framework}) strategy still aims to merge the scalability of training neural-network models with the interpretability offered by SR. However, instead of deducing a model consistent with just the data, we establish a theoretical link between this computational approach and the classical discovery of a polyconvex energy functional for isotropic hyperelasticity, which features an additive structure, as stated in Lemma \ref{lemma:addetivInvars}.

\begin{figure}[h!]
  \centering
  \includegraphics[width=0.8\textwidth]{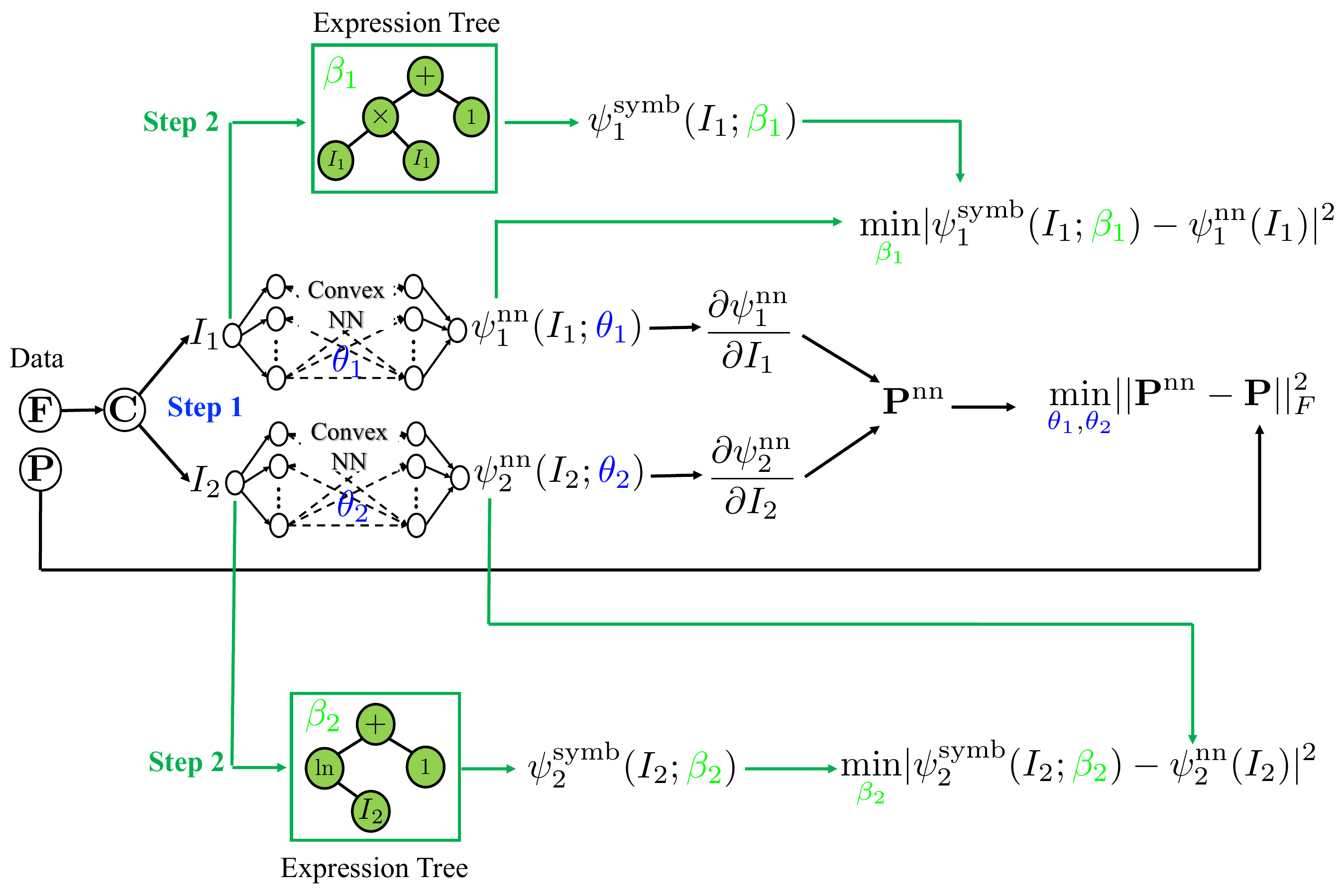}
  \caption{Proposed interpretable hyperelastic model discovery via a two-step approach. In \textcolor{blue}{step 1}, a PNAM is trained using a \textcolor{blue}{scalable} gradient-based optimizer with {multivariate} hyperelasticity data. In \textcolor{green}{step 2}, an evolutionary-based optimizer finds \textcolor{green}{interpretable} symbolic equations by conducting SR over {univariate} functions in {parallel}. The data are assumed to be available in the form strain-stress pair $(\tensor{F}, \tensor{P})$. Trainable neural network parameters are shown in blue, i.e., $\vec{\theta}_1, \vec{\theta}_2$. The unknown structure of the expression trees are shown in green, i.e., $\vec{\beta}_1, \vec{\beta}_2$, which are found by an evolutionary-based algorithm.
  }
  \label{fig:framework}
\end{figure}

\subsection{Polyconvex Neural Additive Model}
\label{sec:ICNN}

Here, we introduce the polyconvex neural additive model (PNAM). Each shape function $\psi_i(I_i)$ in the PNAM is parameterized by a separate neural network which must be convex, positive, and non-decreasing. A particular class of neural networks known as input-convex neural networks (ICNNs) designed by \citep{amos2017input} fulfills all these properties.

The forward operations $\psi^{ICNN}(I; \vec{\theta})$ in an ICNN are as follows,
\begin{align}
    &\vec{h}_1 = g_0(I \vec{v}^0 + \vec{b}^0); \quad  \vec{v}^0, \vec{b}^0, \vec{h}^1 \in \mathbb{R}^{l_1}, 
    \label{eq:icnnStep1}
    \\
    &\vec{h}_{j+1} = g_j(\tensor{W}^j \vec{h}_j + I \vec{v}^j + \vec{b}^j);\quad  \vec{v}^j, \vec{b}^j, \vec{h}_j \in \mathbb{R}^{l_{j+1}}, \tensor{W}^j\in {\mathbb{R}^{+}}^{l_j \times l_{j+1}}, 1\le j < L,
    \label{eq:icnnStep2}
    \\
    &\psi^{\text{ICNN}} = h_L,
\end{align}
where, as shown in Figure \ref{fig:icnn},  L is the number of hidden layers, $g_j$ are activation functions for each layer, $h_j$ are the hidden representation of data at each layer, and $\tensor{W}^j$, $\vec{b}^j$, and $\vec{v}^j$ are unknown, trainable parameters of the neural network. 

\begin{figure}[h]
  \centering
  \includegraphics[width=0.5\textwidth]{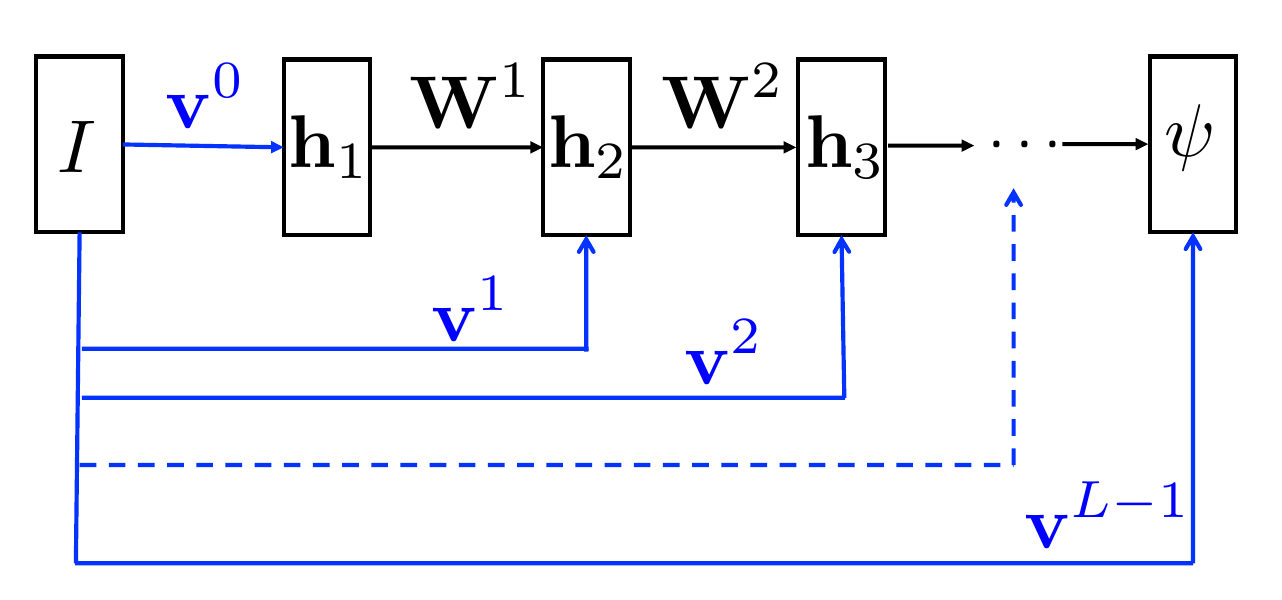}
  \caption{Input-convex neural network which is guaranteed to be convex with respect to the input $I$, non-decreasing, and non-negative by construction.}
  \label{fig:icnn}
\end{figure}

\begin{proposition}
 The function $\psi^{\text{ICNN}}(I, \vec{\theta})$ is convex with respect to the input $I$ if all $\tensor{W}^j$ are non-negative, and all activation functions $g_j$ are convex.
\end{proposition}
\begin{proof}
For the sake of completeness, we provide a straightforward proof. The first step in Eq.~\eqref{eq:icnnStep1} establishes the convexity of each component in the vector $\vec{h_1}$ (see Lemmas \ref{lemma:affineConvx} and \ref{lemma:composConvx}. ) To verify the proposition, we must demonstrate the convexity of each component of $\vec{h}_{j+1}$ with respect to $I$ for $j\ge 1$.  Equation \eqref{eq:icnnStep2} confirms this result, as each component of the resulting vector $\tensor{W}^j \vec{h}_j$ is convex (see Lemma \ref{lemma:posSumConvx}).
\end{proof}

\begin{lemma}\label{lemma:posSumConvx}
A weighted sum of convex functions $h_i(x)$ with non-negative wights $w_i$ results in a convex function in $x\in \mathbb{R}$.
\end{lemma}

\begin{proof}
Let us define $g(x) = w_i h_i(x)$ with constants $w_i \in \mathbb{R}^{+}$ and convex functions $h_i$, for any arbitrary $x_1, x_2,  \text{and } \lambda$ we have,  with summation over $i$,
\begin{align}
	g(\lambda x_1 + (1-\lambda)x_2) &= 
	w_i h_i(\lambda x_1 + (1-\lambda)x_2)
	\\
	&\le 
	w_i
	\left[
		\lambda h_i(x_1) + (1-\lambda) h_i(x_2)
	\right]
	\quad \text{via convexity of } h_i \text{ and positivity of } w_i
	\\
	&=
		\lambda w_i h_i(x_1) + (1-\lambda) w_i h_i(x_2)
	\\
	&=
	\lambda g(x_1) + (1-\lambda) g(x_2).
\end{align}
\end{proof}

\remark{
In this study, we utilize the softplus and softplus2 activation functions, both of which are convex, smooth, \alert{and non-decreasing}. These functions are defined as follows,
\begin{align}
&\text{softplus}(x) = \ln(1 + \exp(x)),
\\
&\text{softplus2}(x) = \text{softplus}(x)^2.
\end{align}
Previous works \citep{klein2022polyconvex, as2022mechanics, kalina2023fe} demonstrate the application of these activation functions in the context of ICNN modeling of hyperelastic energy functionals.
}

\subsection{Objective Function For Neural Network Training}
Each shape function $\psi_i(I_i)$ in Eq.~\eqref{eq:final} is parametrized by a different ICNN $\psi^{ICNN}_i(I_i; \vec{\theta}^i)$ where each $\vec{\theta}^i$ concatenates all trainable parameters associated with the i-th ICNN. In this work, we assume we only have access to strain and stress measurements stored in a dataset $\mathcal{D} = \{\tensor{F}^k, \tensor{P}^k\}_{k=1}^{N_{\text{data}}}$ for $N_{\text{data}}$ number of measurements. Notice that in realistic experiments, access to values of the energy functional or its Hessian is not possible although one may obtain them from numerical simulations. 

To calibrate parameters $\vec{\theta} = \{\vec{\theta}^i\}_{i=1}^2$, we minimize the following loss function,
\begin{equation}
    \vec{\theta} = \underset{\vec{\theta}}{\argmin} \sum_{k=1}^{N_{\text{data}}} \big|\big|
    \parder{\psi^{\text{tot}}(I_1, I_2; \theta)}{\tensor{F}}|_{\tensor{F}^{k}}
    -
    \tensor{P}^{k}
    \big|\big|_F^2,
\end{equation}
where $||\cdot||_F$ is the Frobenius norm for a second-order tensor. We use the ADAM algorithm \citep{kingma2014adam} to minimize this single objective function.

\subsection{Symbolic Regression}
\label{sec:symbReg}
Symbolic regression (SR) aims to identify a mathematical expression that optimally fits a provided dataset without predefining the expression's form. 
The set of possible expressions is typically defined by establishing a range of mathematical operators, functions, variables, and constants, which are then efficiently represented using binary trees (refer to Figure \ref{fig:bnary_tree}). Genetic programming is a widely adopted stochastic optimization technique for exploring the combinatorial space encompassing all conceivable mathematical expressions \citep{koza1994genetic, schmidt2009distilling, wang2019symbolic}. Moreover, recent advancements have been made utilizing deep reinforcement learning methods, offering alternative and efficient means for discrete searches within the domain of tree data structures \citep{petersen2019deep, landajuela2021discovering}.

\begin{figure}[h]
  \centering
  \includegraphics[width=0.5\textwidth]{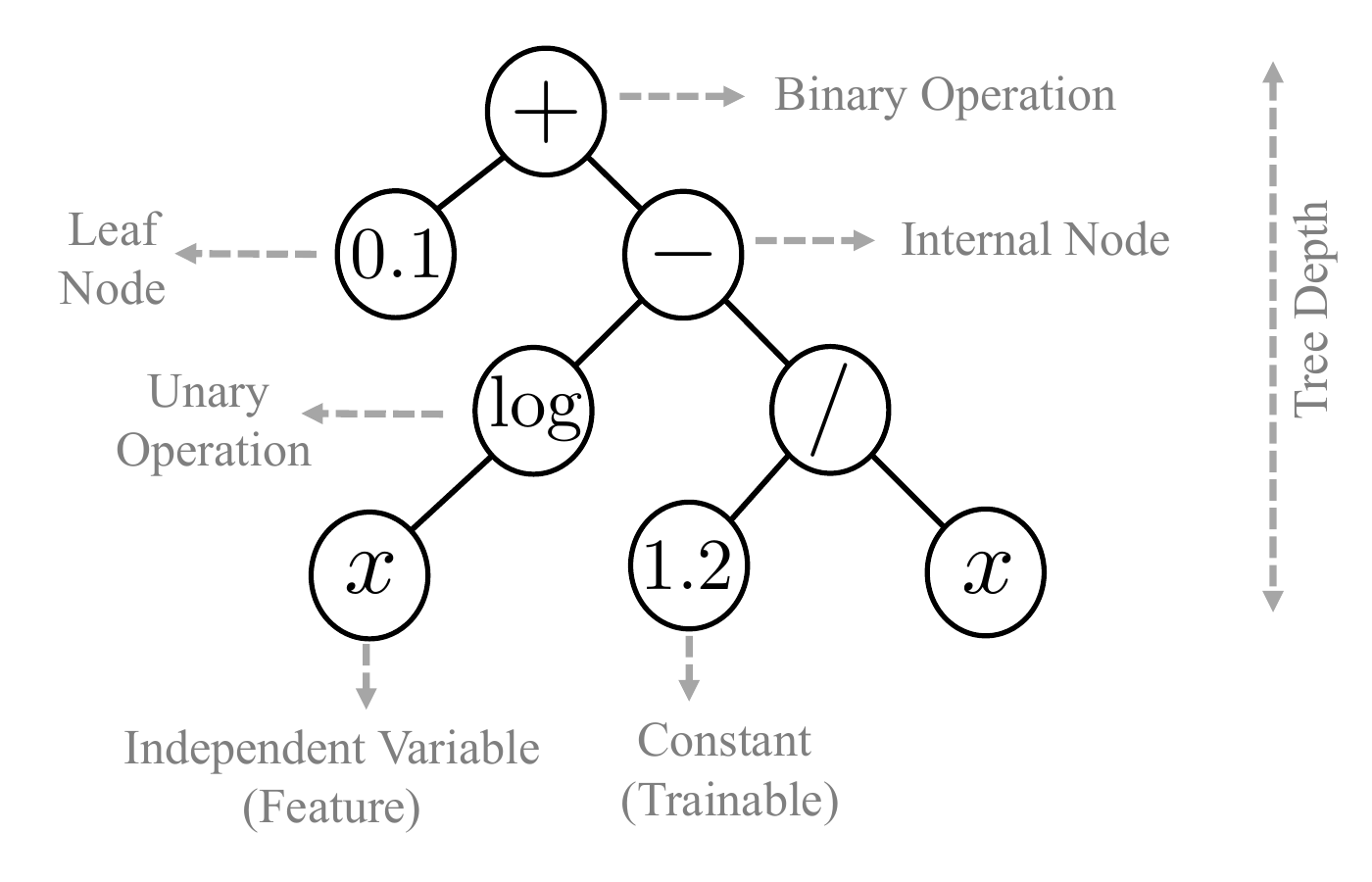}
  \caption{Binary expression tree for equation representation.  The program $\log(x) - \frac{1.2}{x} + 0.1$ is represented by an expression tree with a depth 3 and size of 8 (the total number of nodes).
  }
  \label{fig:bnary_tree}
\end{figure}

Genetic programming employs a population of candidate solutions, randomly generated at the algorithm's start, to represent the space of potential expressions. Each individual candidate solution is described as a binary expression tree (refer to Figure \ref{fig:bnary_tree}), where the leaves symbolize input variables or constants, and the internal nodes denote mathematical operations or functions. The algorithm evaluates the fitness of each candidate solution by comparing its output values to the target values. In this study, we use the mean square error (MSE) as the fitness measure, i.e., 
\begin{equation}
    \text{MSE}(y^{\text{model}}, y^{\text{data}})
    = \frac{1}{N_{\text{data}}} \sum_{i=1}^{N_{\text{data}}} 
    (y^{\text{model}}_i - y^{\text{data}}_i)^2.
\end{equation}

The genetic programming algorithm utilizes an iterative process that evolves the population of candidate solutions through three main operations: selection, crossover, and mutation. This process draws inspiration from natural selection. In the selection step, higher-performing individuals from the current population are chosen based on their fitness scores. Crossover (shown in Figure \ref{fig:cross_mut_ops}(a)) combines the genetic information of two individuals to generate offspring with characteristics inherited from both parents. On the other hand, mutation (depicted in Figure \ref{fig:cross_mut_ops}(b)) involves randomly altering some genetic components of the individual, thereby introducing new variations into the population.

\begin{figure}[h!]
\centering
\subfigure[]
{\includegraphics[height=0.35\textwidth]{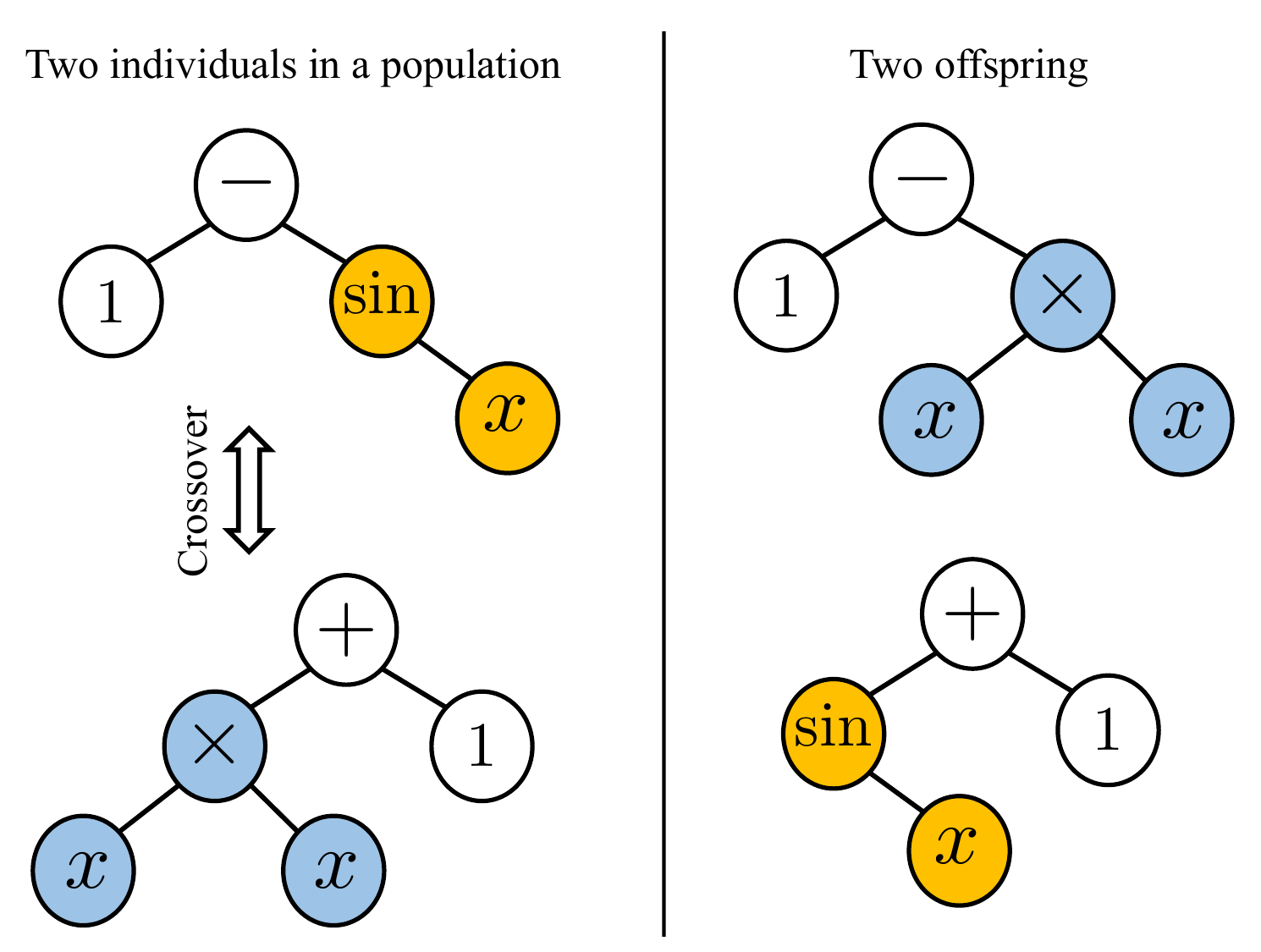}}
\subfigure[]
{\includegraphics[height=0.35\textwidth]{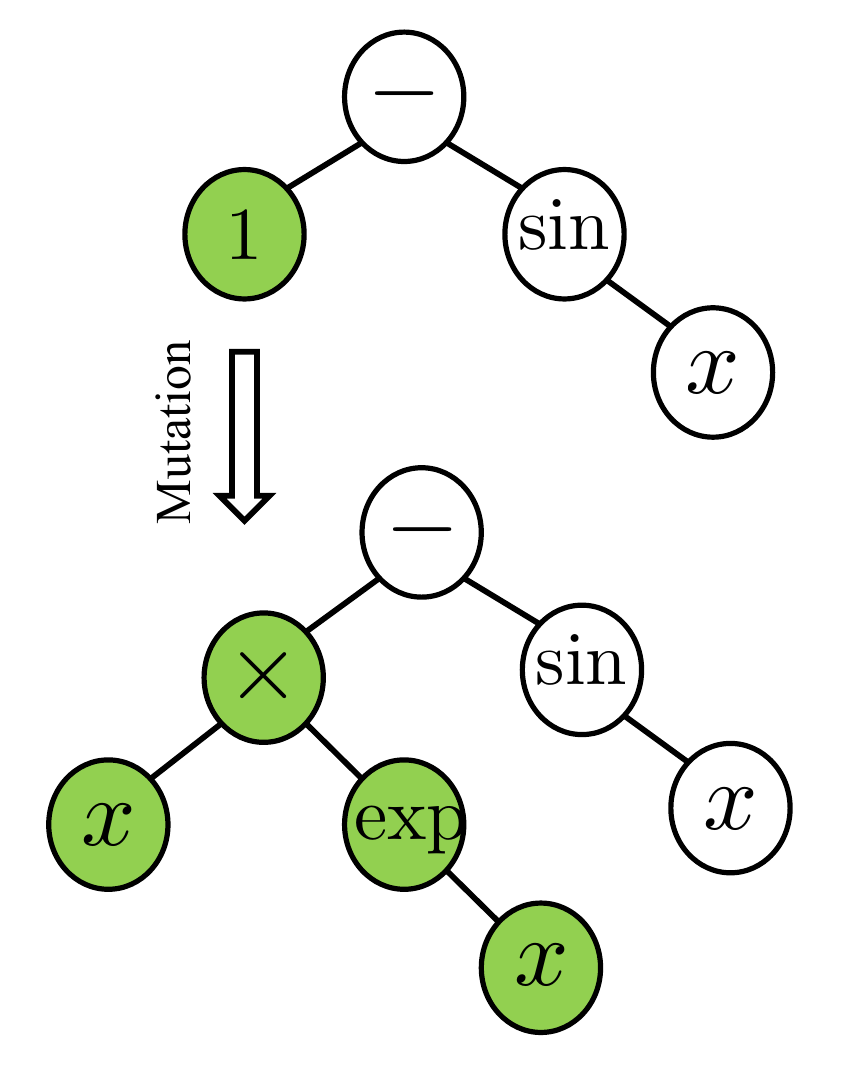}}
  \caption{
  (a) crossover and (b) mutation operations in an evolutionary-based SR algorithm.
  }
\label{fig:cross_mut_ops}
\end{figure}

The genetic programming algorithm creates improved candidate solutions iteratively until it finds a satisfactory mathematical expression that fits the data well. During SR optimization, there is a tradeoff between model complexity and expressivity. Striking a balance between the complexity and expressivity of the discovered analytical expression is vital to enhance interpretability and mitigate overfitting. Users can choose the best equation for their requirements, considering desired accuracy and simplicity.

\definition[Complexity score in SR]{The assessment of a symbolic equation's complexity is often qualitative due to the lack of a universally agreed-upon definition. In this study, we adopt the complexity measure proposed in \cite{cranmer2023interpretable}, which uses the number of nodes in the expression tree as the complexity score. While assigning varying weights to different node types is feasible, such as considering $\exp(\cdot)$ as more complex than $+$, we refrain from incorporating such weightings in our analysis.
}

Symbolic equations used during inference have a smaller memory footprint compared to neural networks, making them more portable.  Research has shown that equations discovered by SR generalize well beyond the training data support \citep{kim2020integration}. However, when dealing with multi-dimensional vector-valued or tensor-valued functions, SR becomes notably more challenging due to the complex combinatorial nature of the optimization problem involved in searching for the appropriate mathematical expression \citep{icke2013improving}. Our proposed method overcomes this concern, as mentioned earlier, by extracting a symbolic equation for each single-variable to single-variable ICNN submodel. This feature facilitates parallel computing, allowing SR algorithms to run simultaneously without any issues.

\section{Reduced forms for experimental calibration}
\label{sec:expReduce}
Motivated by real experimental setups, we calibrate model parameters for uniaxial extension (UE), equibiaxial extension (EBE), and pure shear (PS) loading conditions in the incompressible limit. In this section, we summarize reduced forms of equations for a general incompressible, isotropic free energy functional $\psi(I_1, I_2)$ associated with each of these boundary conditions.

\begin{figure}[h]
  \centering
  \includegraphics[width=0.8\textwidth]{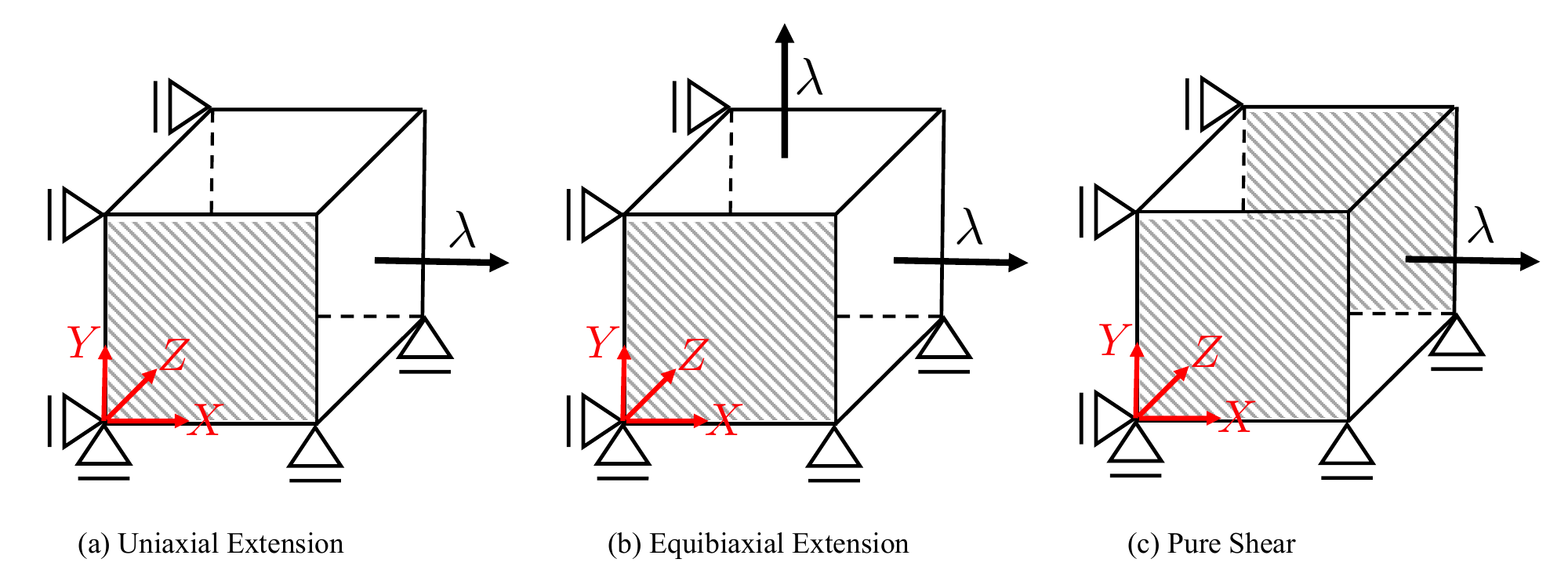}
  \caption{Different loading conditions for model calibration and training.}
  \label{fig:loading}
\end{figure}

\subsection{Uniaxial Extension}
The easiest experimental setup is UE where only one side of the cubic sample is deformed along its normal axis, e.g., the $X$-axis as shown in Figure~\ref{fig:loading}(a). Based on this condition, we only control the stretch $\lambda_X = \lambda$, but, due to the symmetry induced by isotropy, the other two stretches are equal, and, hence the incompressibility condition leads to $\lambda_Y = \lambda_Z = \lambda^{-0.5}$. Under these boundary conditions, only $P_{11}$ becomes non-zero which can be found in Eqs. \eqref{eq:incomp} and \eqref{eq:final}. In summary, the following reduced form is obtained for the UE loading condition,
\begin{align}
    &\tensor{F} = \text{diag}(\lambda, \lambda^{-0.5}, \lambda^{-0.5}),\\
    &\tensor{P} = \text{diag}(P_1, 0, 0),\\
    &P_1 = 2 
    \left(
        \parder{\psi}{I_1} + \frac{1}{\lambda}\parder{\psi}{I_2}
    \right)
    \left(
        \lambda - \frac{1}{\lambda^2}
    \right).
\end{align}

\subsection{Equibiaxial Extension}
EBE is similar to UE where the extension loading is applied over two faces of the cubic sample as depicted in Figure~\ref{fig:loading}(b). Based on this condition, we control the stretches $\lambda_X = \lambda_Y = \lambda$, while the other stretch is constrained to $\lambda_Z = \lambda^{-2}$, due to incompressibility. Under these boundary conditions, only $P_{11}$ and $P_{22}$ become non-zero and the same because of the symmetry induced by isotropy. In summary, the following reduced form is obtained for EBE loading condition,
\begin{align}
    &\tensor{F} = \text{diag}(\lambda, \lambda, \lambda^{-2}),\\
    &\tensor{P} = \text{diag}(P_1, P_1, 0),\\
    &P_1 = 2 
    \left(
        \parder{\psi}{I_1} + \lambda^2 \parder{\psi}{I_2}
    \right)
    \left(
        \lambda - \frac{1}{\lambda^5}
    \right).
\end{align}

\subsection{Pure Shear}
PS loading condition can be achieved by restricting UE from any movement in the $Z$-axis, i.e., $\lambda_Z=1$ (see Figure~\ref{fig:loading}(c).) We control the stretch $\lambda_X = \lambda$, while the other stretch is constrained to $\lambda_Y = \lambda^{-1}$. Under these boundary conditions, only $P_{11}$ and $P_{33}$ become non-zero. In summary, the following reduced form is obtained for PS loading condition,
\begin{align}
    &\tensor{F} = \text{diag}(\lambda, \lambda^{-1}, 1)\\
    &\tensor{P} = \text{diag}(P_1,0, P_3),\\
    &P_1 = 2 
    \left(
        \parder{\psi}{I_1} + \parder{\psi}{I_2}
    \right)
    \left(
        \lambda - \frac{1}{\lambda^3}
    \right),\\
    &P_3 = 2 
    \left(
        \parder{\psi}{I_1} + \lambda^2 \parder{\psi}{I_2}
    \right)
    \left(
        1 - \frac{1}{\lambda^2}
    \right).
\end{align}

\section{Numerical Examples}
\label{sec:numExamp}
We demonstrate the application of our proposed method in discovering hyperelastic constitutive models through two numerical examples. Firstly, in Section \ref{sec:treloarData}, we utilize experimental data from the literature to identify a material model. We compare the simplicity and accuracy of our discovered model with recent efforts in the field. Additionally, we assess and discuss the extrapolation capabilities of our method compared to a standard neural network approach.

In the second problem, in Section \ref{sec:compositeProb}, we showcase the method's potential for modeling more complex materials. To showcase this ability, we generate virtual experimental data using finite element simulations for a randomly synthesized particle-reinforced composite. 

\remark{
Our framework is implemented using \texttt{PyTorch} \citep{paszke2019pytorch} and \texttt{PySR} \citep{cranmer2023interpretable} open-source packages for neural network and SR training, respectively. 
}

\subsection{Treloar's Data for Vulcanized Rubber}
\label{sec:treloarData}
In this example, we demonstrate the application of the proposed scheme in finding a hyperelastic energy functional for the well-known Treloar's experimental data for vulcanized rubber \citep{treloar1944stress}. We then compare the accuracy and simplicity of the discovered model with a reference model. Additionally, we highlight the significance of physics augmentation and interpretability for achieving better generalization compared to black-box neural network models.

We use 90\% of Treloar's data to train the model discovery algorithm and the remaining 10\% to validate. Training-related hyperparameters and setups are discussed in Appendix \ref{sec:paramTrel}. The proposed algorithm discovers the following equation,
\begin{equation}
    \psi = 
    c_{11} I_1 + c_{12}\exp{(c_{13}I_1)}
    +
    c_{21} I_2 + c_0,
    \label{eq:pr1-us}
\end{equation}
where $c_{11}=0.1502$, $c_{12}=0.0771$, $c_{13}=0.0665$, $c_{21}=0.0035$, and $c_0 = 261.5397$. In Table \ref{tab:example}, we compare its accuracy with another equation found by a recently introduced SR algorithm in \cite{abdusalamov2023automatic},
\begin{equation}
    \psi_{\text{ref}} = \sqrt{
    0.93296\exp{(0.08031 I_1)}
    +
    \sqrt{I_1 - 0.080316}
    +
    (0.0232113 I_1 + 0.021633)I_1
    }.
    \label{eq:pr1-itskov}
\end{equation}
In this table, the $R^2$-score is chosen as the measure of accuracy over the entire data.  The $R^2$-score between predicted values $\{y_i\}_{i=1}^{N_{\text{data}}}$ and ground-truth values $\{y^*_{i}\}_{i=1}^{N_{\text{data}}}$ is defined as follows,
\begin{equation}
    R^2(y^*, y) = 1 - 
    \frac{
    \sum_{i=1}^{N_{\text{data}}} (y^*_i - y_i)^2
    }
    {
    \sum_{i=1}^{N_{\text{data}}} (y^*_i - \text{mean}(y^*))^2
    }.
\end{equation}
Based on this measure, our discovered equation shows improved performance in the EBE dataset, with slight enhancements observed in the UE and PS datasets as well. Considering the depth of the expression tree and square root operators as more complex, our discovered equation is qualitatively simpler.  Additionally, in our work, we use 90\% of the data to find the equation, whereas the reference utilizes the entire dataset.
\begin{table}[h]
  \centering
  \caption{Comparing $R^2$-scores of our discovered equation and the reference equation for Treloar's data.}
  \label{tab:example}
  \begin{tabular}{c c c c}
    \hline
    Equation & Uniaxial Extension & Equibiaxial Extension & Pure Shear \\
    \hline
    This work (Eq.~\eqref{eq:pr1-us}) & 0.99\textcolor{green}{83} & 0.\textcolor{green}{9904} & 0.9\textcolor{green}{963} \\
    \hline
    Reference (Eq.~\eqref{eq:pr1-itskov}) & 0.9979 & 0.8692 & 0.9887 \\
  \end{tabular}
\end{table}

The trained ICNN model $\psi_1(I_1)$ for the PNAM solution $\psi(I_1, I_2) = \psi_1(I_1) + \psi_2(I_2)$ is depicted in Figure \ref{fig:psi1Symb} by the orange curve. The result for $\psi_2(I_2)$ is omitted for brevity as it exhibits characteristics of a linear function, making it straightforward to handle. As mentioned earlier, in step 2 of the proposed method, we apply the SR algorithm to extract an equivalent symbolic equation for the learned model. The SR training process takes less than 3 minutes on a MacBook Pro laptop with a quad-core Intel Core i5 processor running at 1.4 GHz and using 16 GB of RAM. The SR algorithm lists several candidates ranked by their accuracy/simplicity, with more accurate candidates being less simple in forms. 

In Figure \ref{fig:psi1ComplAcc}, we present accuracy vs. complexity plots for both the function $\psi_1(I_1)$ and its gradient $\psi_1^\prime(I_1)$. While the gradient information is not employed in the fitness measure, we recommend considering this information for model selection during the post-processing step among all available candidates. This consideration is crucial because the accuracy of the gradient directly impacts the final stress predictions. We employ these selection criteria based on two factors: candidates must (1) have errors below a specified threshold and (2) exhibit an abrupt accuracy improvement compared to their simpler counterparts (see the red vertical line in Figure \ref{fig:psi1ComplAcc}.) A comparison between the selected equation and the ICNN model is illustrated in Figure \ref{fig:psi1Symb}, with the symbolic equation depicted by the blue curve. A good agreement between the ICNN model and the symbolic equation is observed. Moreover, the behavior of the symbolic function beyond the ICNN training data appears to be reasonable.

\begin{figure}[h]
  \centering
  \includegraphics[width=0.9\textwidth]{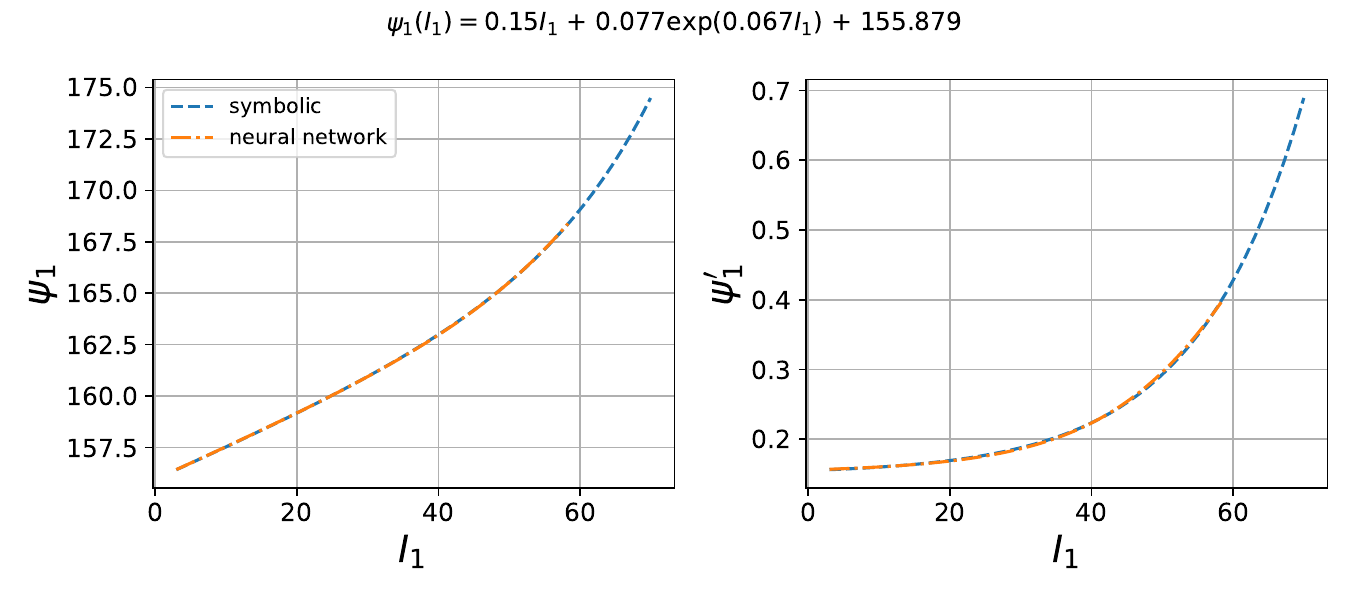}
  \caption{Discovered symbolic function from the trained neural network representation. The discovered function and its gradient are shown in left and right plots, respectively. Only data within the  training range (orange curves) are used for the SR.
  }
  \label{fig:psi1Symb}
\end{figure}

\begin{figure}[h]
  \centering
  \includegraphics[width=0.9\textwidth]{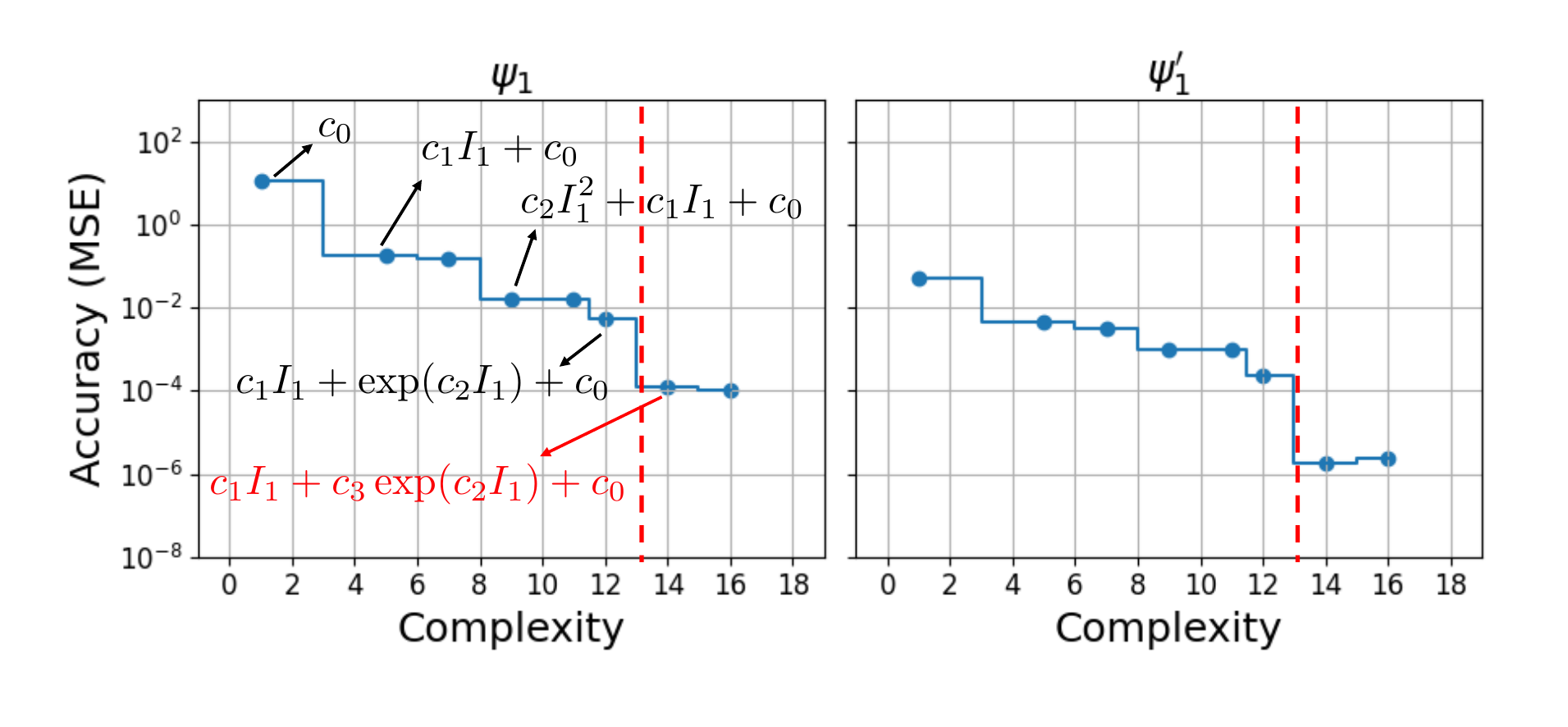}
  \caption{
  Complexity (parsimony) vs. accuracy of the discovered function and its gradient for Treloar's data. More accurate equations  tend to be more complex (i.e., contain more terms.) Achieving a balance between accuracy and simplicity is essential for both the discovered function and its gradient.
  }
  \label{fig:psi1ComplAcc}
\end{figure}

We further compare the proposed scheme with the solution obtained from a vanilla ANN. In the case of the vanilla ANN, we substitute $\psi(I_1, I_2)$ in Eq.~\eqref{eq:final} with a conventional multilayer perceptron (MLP), accepting both $I_1$ and $I_2$ in its input layer. The model accuracy during the optimization iterations in step 1 for the PNAM and vanilla ANN is illustrated in Figure \ref{fig:pr1-loss} for both the training and validation data. While the final loss value of the training curve for the vanilla ANN is lower than that for the PNAM, the larger discrepancy between the validation and training curves for the vanilla ANN indicates overfitting. This overfitting may lead to less reliable predictions in the extrapolation regime.

\begin{figure}[h!]
\centering
\subfigure[PNAM]
{\includegraphics[height=0.35\textwidth]{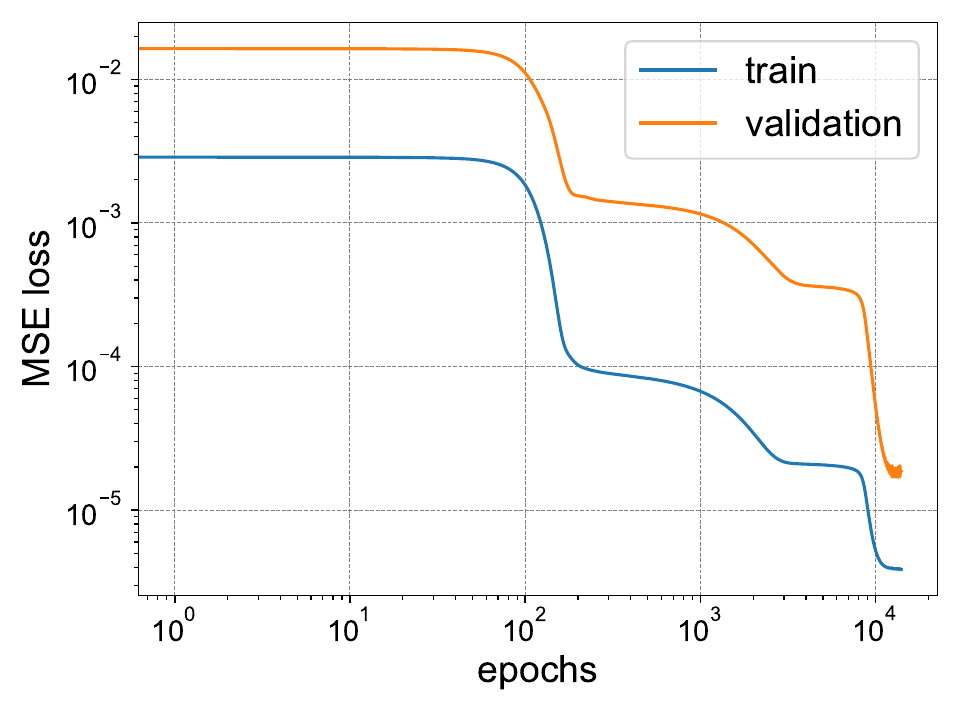}}
\subfigure[Vanilla ANN]
{\includegraphics[height=0.35\textwidth]{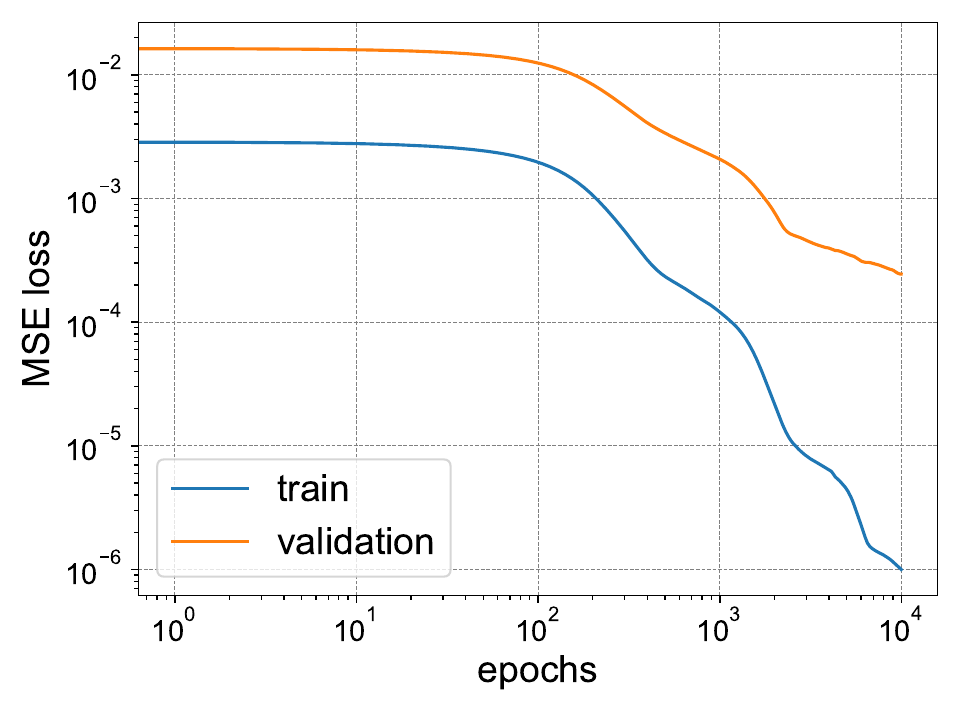}}
  \caption{Training and validation losses for model calibration in step 1 where a neural network solution is found. $90\%$ of the data is utilized for training,  and $10\%$ is utilized for validation.
  }
  \label{fig:pr1-loss}
\end{figure}

To assess both models' capabilities, we compare their $R^2$-scores for all three datasets, as shown in Figure \ref{fig:pr1-predic}. Both models exhibit satisfactory $R^2$-scores across all available data. The PNAM demonstrates better performance for UE and EBE, while the vanilla ANN shows slightly better accuracy in the case of the PS dataset. However, it is essential to note that we observe non-physical behavior for the vanilla ANN in the case of the EBE dataset when moving beyond the training data regime.

\begin{figure}[h!]
\centering
\subfigure[UE Symbolic]
{\includegraphics[height=0.3\textwidth]{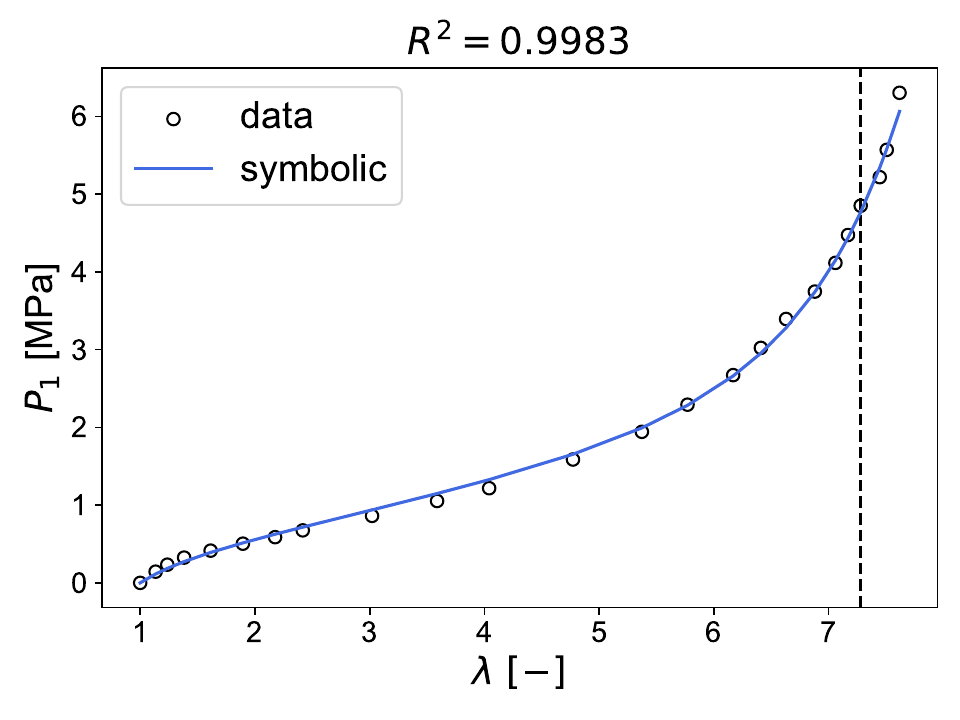}}
\subfigure[UE Vanilla ANN]
{\includegraphics[height=0.3\textwidth]{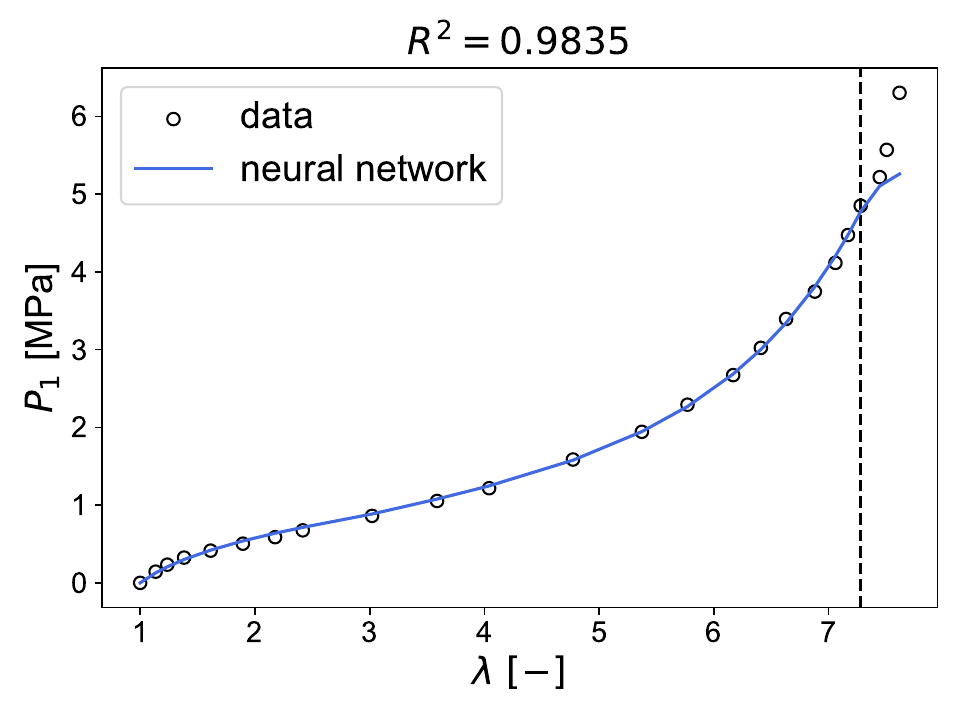}}
\subfigure[EBE Symbolic]
{\includegraphics[height=0.3\textwidth]{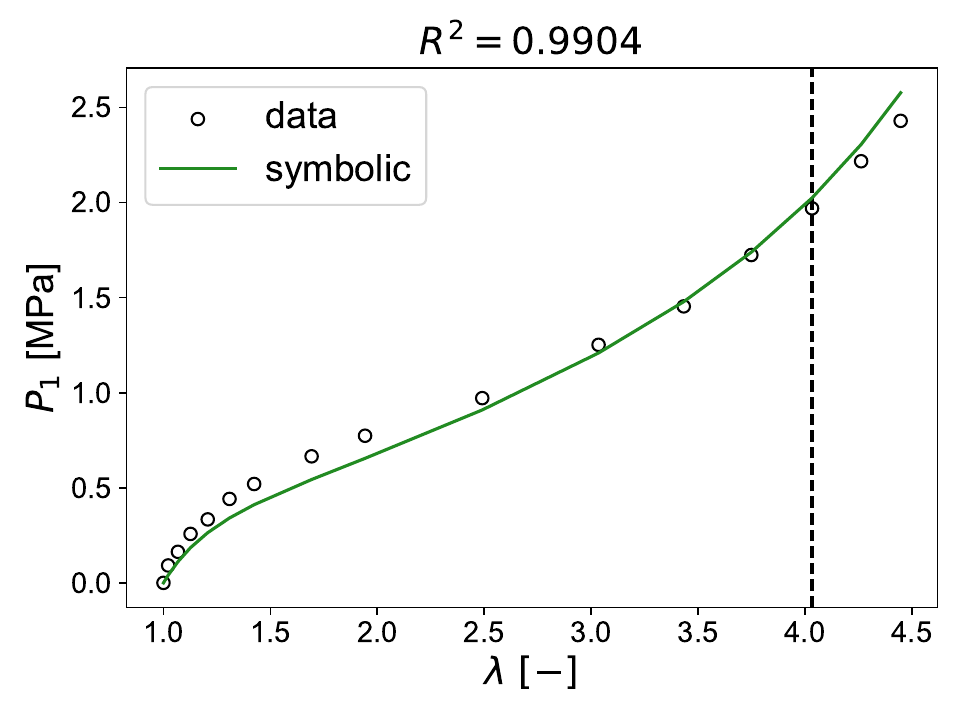}}
\subfigure[EBE Vanilla ANN]
{\includegraphics[height=0.3\textwidth]{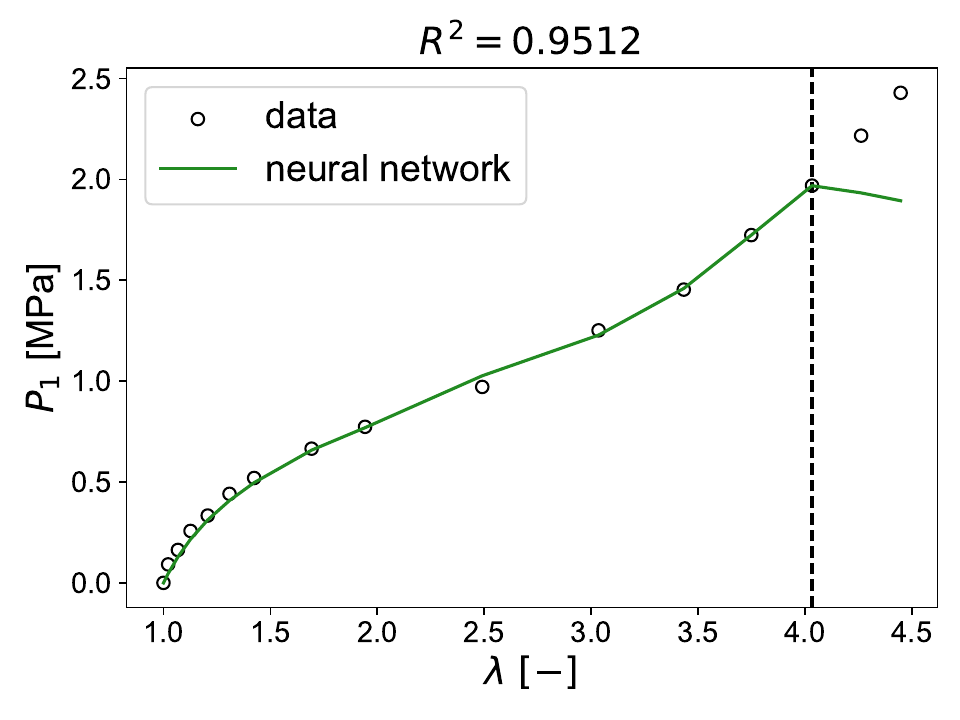}}
\subfigure[PS Symbolic]
{\includegraphics[height=0.3\textwidth]{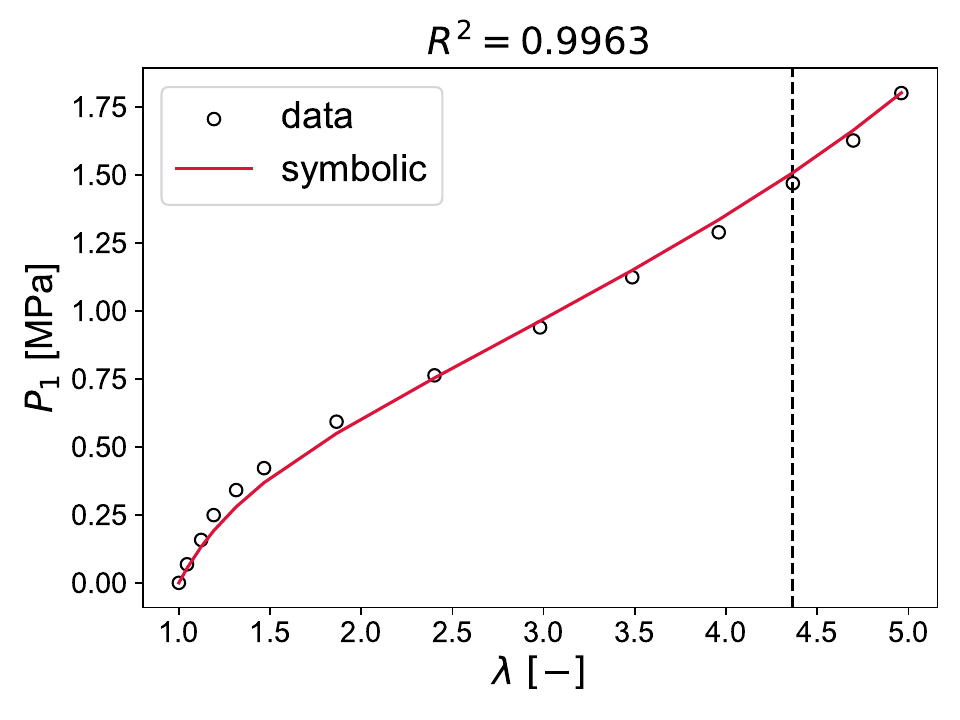}}
\subfigure[PS Vanilla ANN]
{\includegraphics[height=0.3\textwidth]{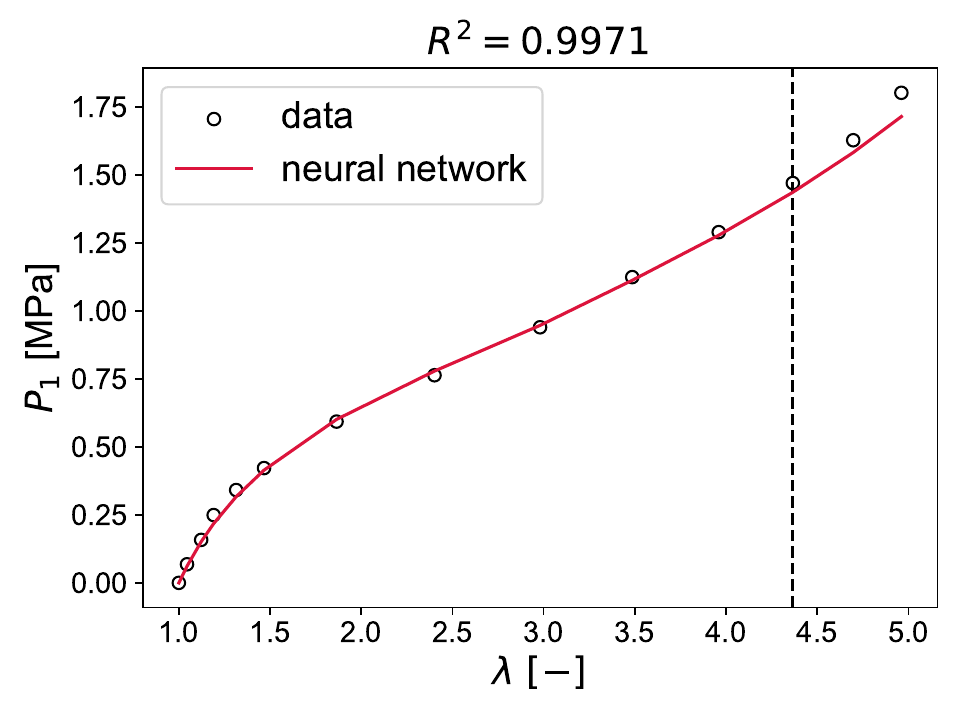}}
    \caption{Predictions of the symbolic equation discovered from the PNAM and a vanilla ANN for Treloar's data. The dashed vertical dashed line shows the distinction between training and validation data, where validation data are to the right of the line. The $R^2$-score indicates the accuracy measure of the entire data (training and validation) for each experimental loading condition. Training is conducted on all three loading conditions at the same time.
    }\label{fig:pr1-predic}
\end{figure}

In order to highlight the extrapolation power resulting from physics augmentation, simplicity, and interpretability of the proposed scheme, we compare predictions of the symbolic equation discovered in this work (Eq.~\eqref{eq:pr1-us}), the one found in the reference (Eq.~\eqref{eq:pr1-itskov}), and the vanilla ANN in Figure \ref{fig:pr1-ref-comp}, where we extend far beyond the training data ranges. Given that we lack access to real data at these ranges, we cannot assess the accuracy of the symbolic equations. However, both equations exhibit reasonable behavior. On the other hand, the response of the vanilla ANN is notably unacceptable, especially under the EBE loading condition.

\begin{figure}[h!]
\centering
\subfigure[Symbolic]
{\includegraphics[height=0.35\textwidth]{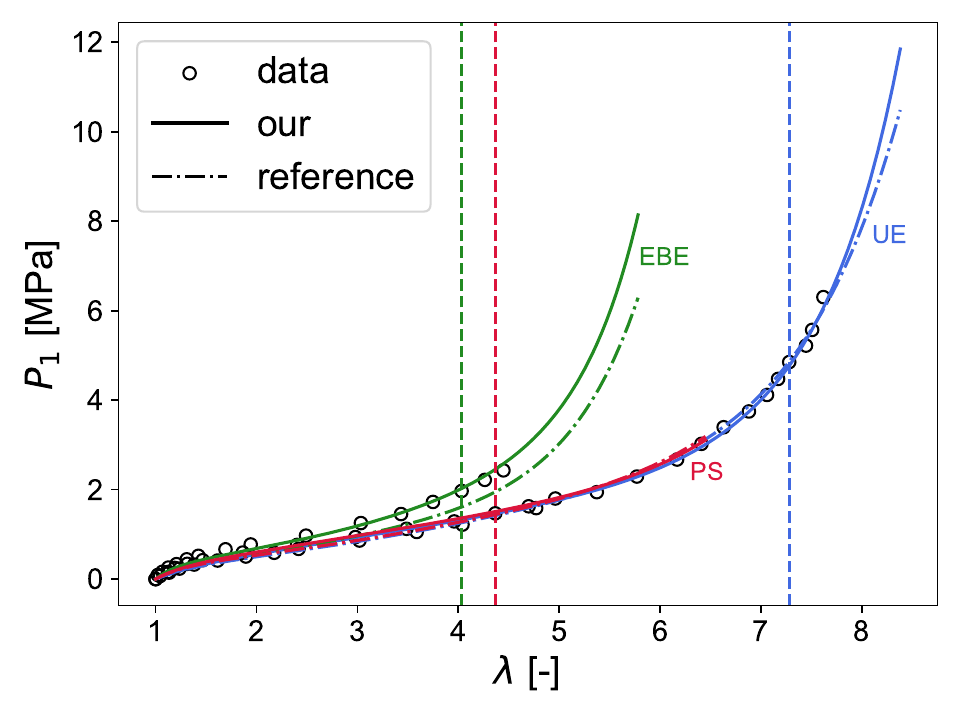}}
\subfigure[Vanilla ANN]
{\includegraphics[height=0.35\textwidth]{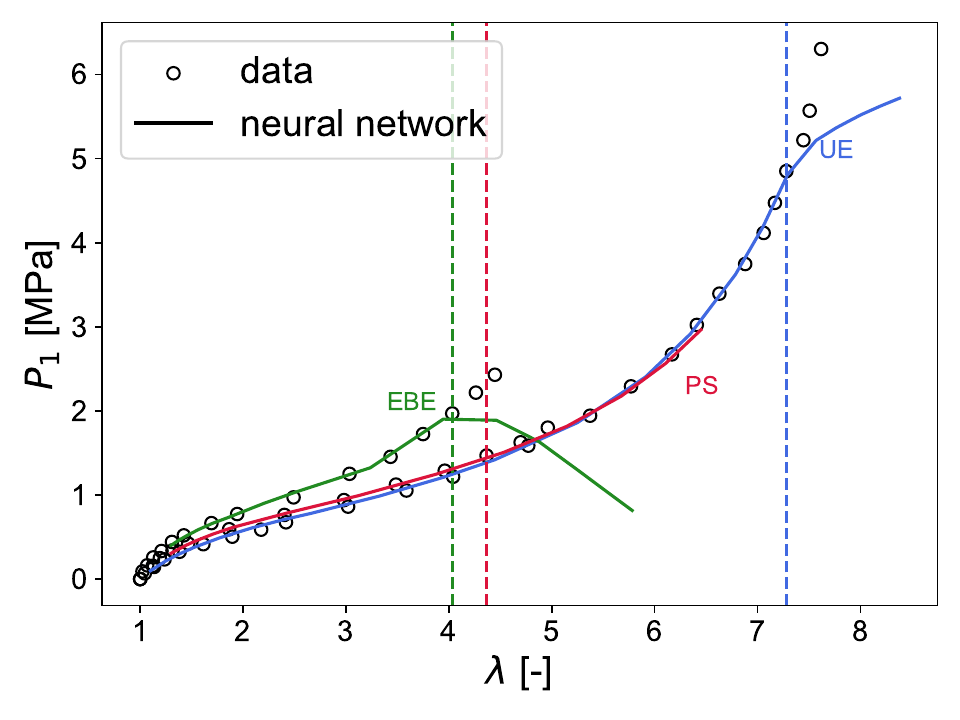}}
  \caption{Extrapolation capabilities of (a) symbolic equations, (b) the vanilla ANN for Treloar's data. In (a), we compare results obtained via Eq.~\eqref{eq:pr1-us} and the reference Eq.~\eqref{eq:pr1-itskov}. The dashed vertical lines indicate the distinction between training and validation data, where validation data are to the right of the lines. Our discovered equation extrapolates well and shows slightly better accuracy than the reference's equation. Predictions of the vanilla artificial neural network are not accurate beyond the training data regimes.
  }
  \label{fig:pr1-ref-comp}
\end{figure}

\subsection{Particle-reinforced Composite}
\label{sec:compositeProb}
In this example, we synthesize data through numerical experimentation on a fabricated particle-reinforced composite. This demonstration highlights the method's efficacy in handling more complex material behaviors. A cubic sample of the composite material with a length of $1$cm is generated by random placements of hard particles in a soft matrix as shown in Figure \ref{fig:compositeMesh}. Neo-Hookean \citep{ogden1997non} and the Exp-Log model \citep{khajehsaeid2013hyperelastic} are assumed for the particles and the matrix, respectively,
\begin{align}
    &\psi_{\text{NH}} = \frac{1}{2} \mu (I_1 - 3),
    \\
    &\psi_{\text{exp-log}} = 
    A \left[
        \frac{1}{a} \exp\left(a(I_1 - 3)\right)
        + b (I_1 - 1) \left(1-\log(I_1 - 2)\right)
        - \frac{1}{a} - b
    \right],
\end{align}
where material parameters $\mu = 0.5$ MPa, $A=0.195$ MPa, $a=0.018$,  and $b=0.33$; in the limit of small deformation $A$ is analogous to the shear modulus $\mu$ in the Neo-Hookean model. The uniaxial behavior of these constituents are shown in Figure \ref{fig:compositeMesh}(b).

\begin{figure}[h]
  \centering
  \includegraphics[width=0.8\textwidth]{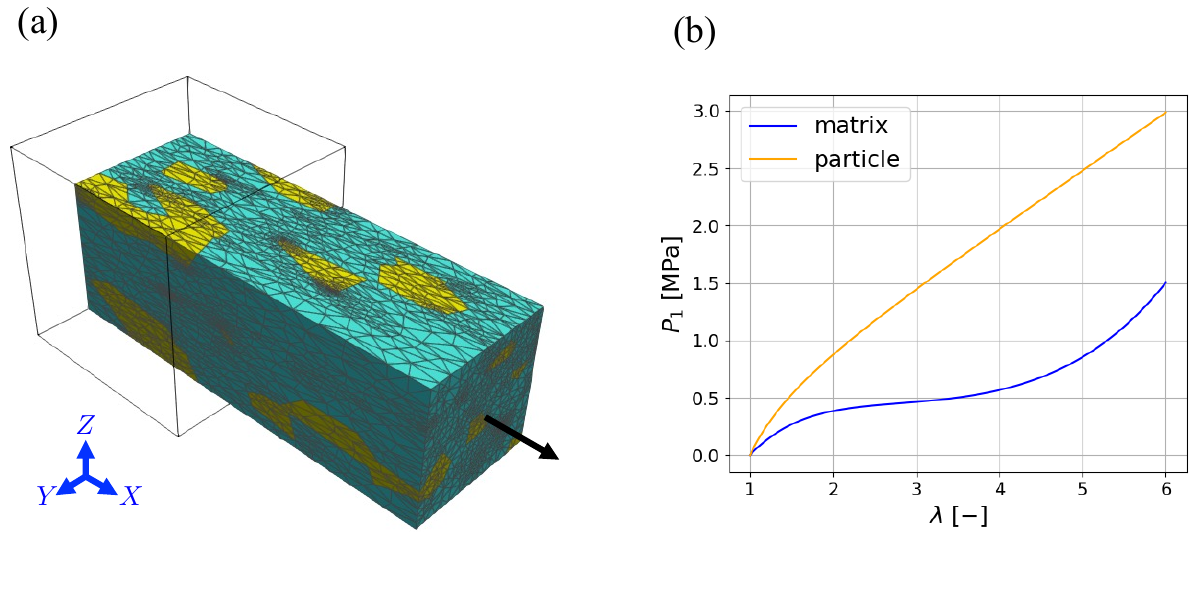}
  \caption{(a) Particle-reinforced composite at 100\% longitudinal strain, and (b) uniaxial behavior of each constituent. In (a), random particles with a total volume fraction of 0.22 are shown in yellow.
  }
  \label{fig:compositeMesh}
\end{figure}

We subject the composite to a uniaxial extension loading condition using a displacement control FEM simulator. During this process, we collect 64 data points with a displacement increment of $0.8$ mm by recording the average forces calculated over the sample surface where the loads are applied; forces are then used to compute stresses. The generated dataset for training the models is depicted in Figure \ref{fig:pr2compare}.

For training purposes, we utilize $75\%$ of the data, represented by green circle points in Figure \ref{fig:pr2compare}. The remaining portion of the data, which entirely lies in the extrapolation regime, is used for validation.  Our method discovers the following equation,
\begin{equation}
    \psi = 
     c_{11}I_1 + c_{12}I_1^4 + \exp(c_{13}I_1)
    +
    c_{21} I_2 + c_{22} I_2^2 \ln(I_2) + c_{0},
\end{equation}
where $c_{11} = 0.06948, c_{12}=1.82532 \times 10^{-6}, c_{13}=-0.05967,c_{21}=0.91519, c_{22}=0.0069682,  \text{and } c_0 = 15.87993107$.
$R^2$-scores of this symbolic equation and the PNAM from which it is discovered are almost the same, with the PNAM demonstrating slightly better performance, as shown at the top of Figure \ref{fig:pr2compare}. Additionally, the comparison between the symbolic representation of each of the learned ICNN models $\psi_i(I_i)$ is presented in Figures \ref{fig:pr2Psi1Symb} and \ref{fig:pr2Psi2Symb}. 
The results indicate good agreement between the neural networks and symbolic equations for both the function values and their gradients.

\begin{figure}[h]
  \centering
  \includegraphics[width=0.6\textwidth]{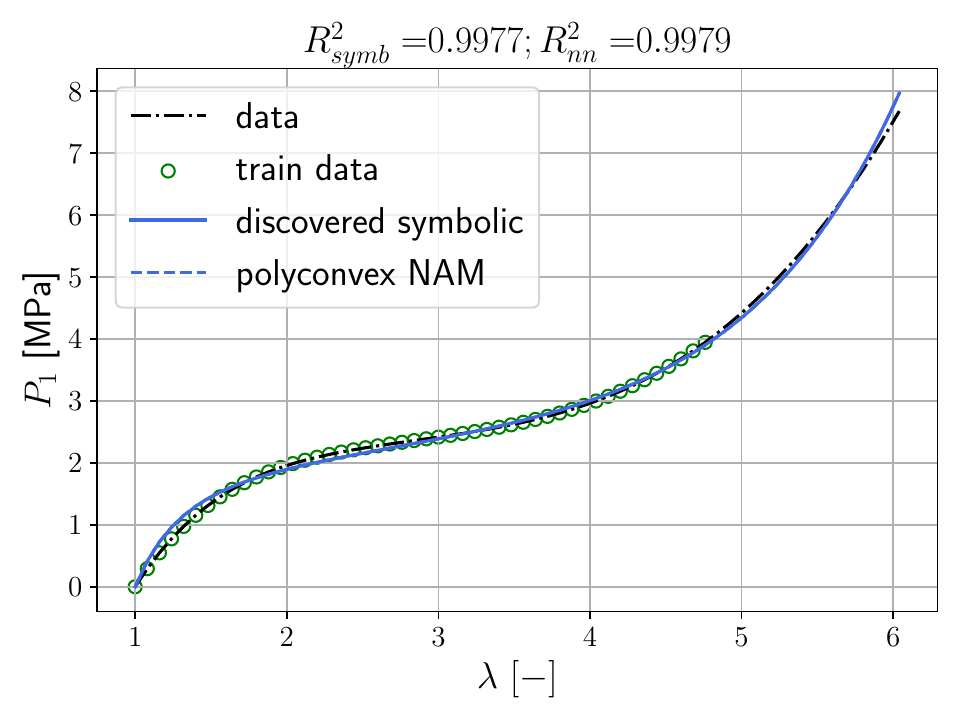}
  \caption{Comparison between the discovered symbolic energy functional and the learned PNAM for particle-reinforced composite data.
  }
  \label{fig:pr2compare}
\end{figure}

\begin{figure}[h]
  \centering
  \includegraphics[width=0.9\textwidth]{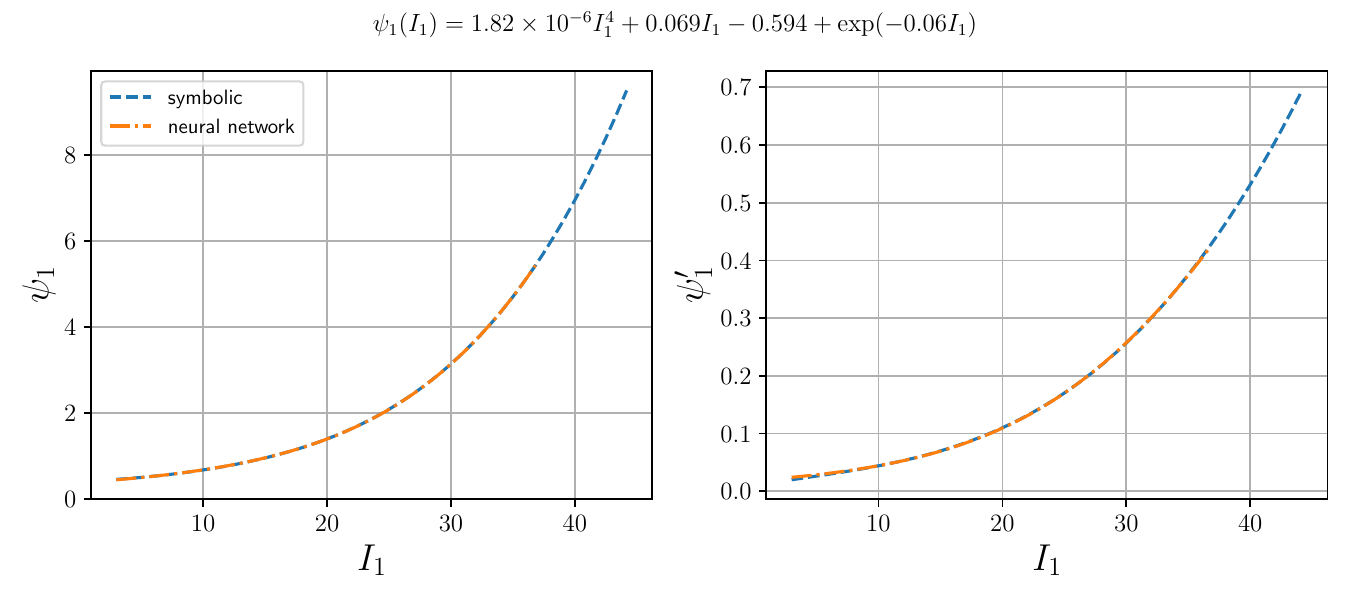}
  \caption{Discovered symbolic function from the trained neural network representation.  The discovered function and its gradient are shown in left and right plots, respectively. Only data in the training range (orange curves) are used for SR.
  }
  \label{fig:pr2Psi1Symb}
\end{figure}

\begin{figure}[h]
  \centering
  \includegraphics[width=0.9\textwidth]{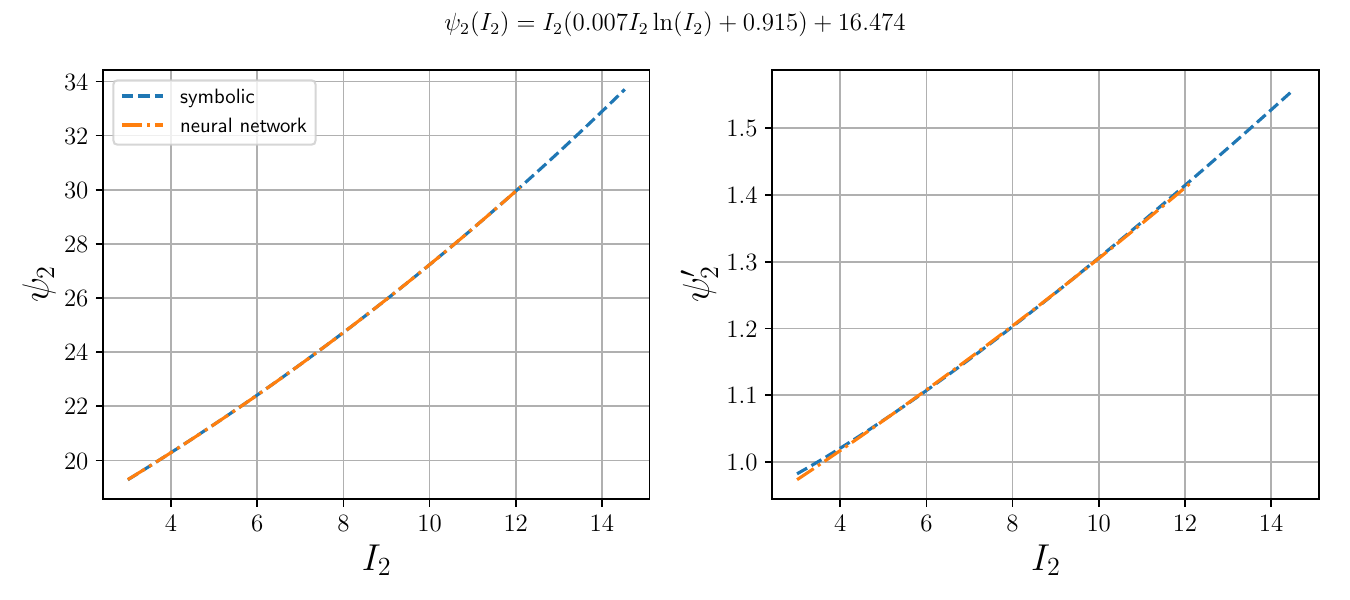}
  \caption{Discovered symbolic function from the trained neural network representation.  The discovered function and its gradient are shown in left and right plots, respectively. Only data in the training range (orange curves) are used for SR.
  }
  \label{fig:pr2Psi2Symb}
\end{figure}

\remark{
In this example, the inclusions are randomly distributed, without a particular geometric pattern or shape, and their volume fraction is relatively small (less than 30\%). Therefore, at the RVE level and under these circumstances, the anisotropic behavior is not significant, allowing us to approximate it with an isotropic model \citep{cho2023large}.  A similar approach to the random generation of inhomogeneities is utilized in other research for the data generation of isotropic models \citep{kalina2022automated}.
}

\subsubsection{Remarks on Model Execution Time}
Another advantage of obtaining symbolic equations is their computational efficiency when compared to their neural network-based counterparts. In Figure \ref{fig:pr2-time-comp}, we compare the execution times for inference between the PNAM and SR. As anticipated, the runtime for symbolic calculations is shorter than that of neural network forward operations, primarily due to the reduced number of floating-point operations in the symbolic equation. While automatic differentiation provides implementation flexibility, it introduces additional overhead due to the construction of the computational graph, which can be avoided if symbolic operations are effectively managed manually.

\begin{figure}[h!]
\centering
\subfigure[strain energy]
{\includegraphics[height=0.35\textwidth]{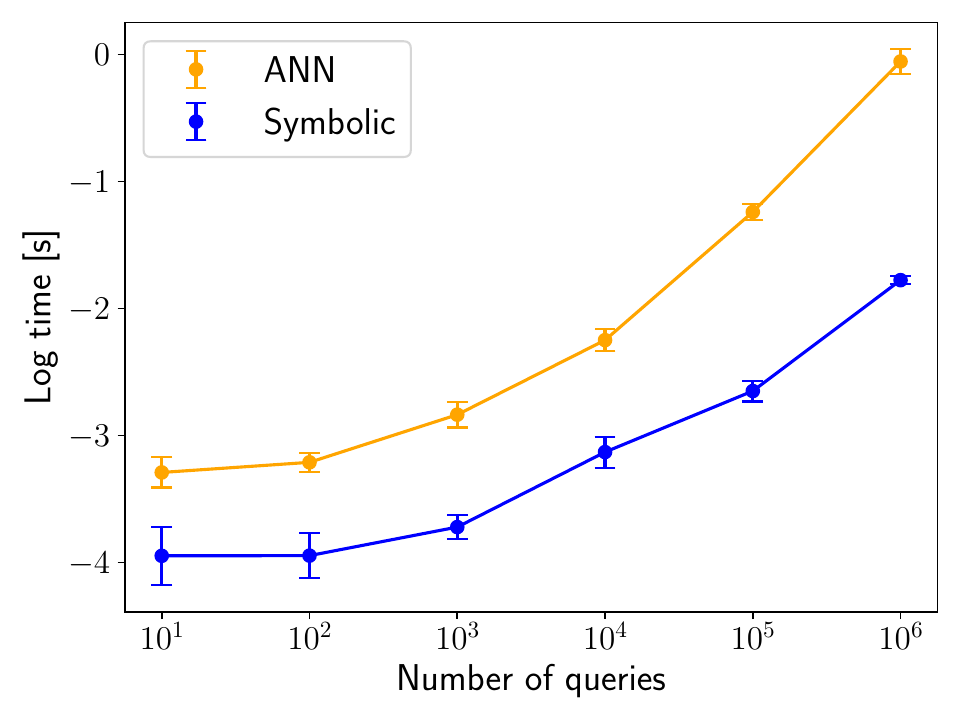}}
\subfigure[gradient of strain energy]
{\includegraphics[height=0.35\textwidth]{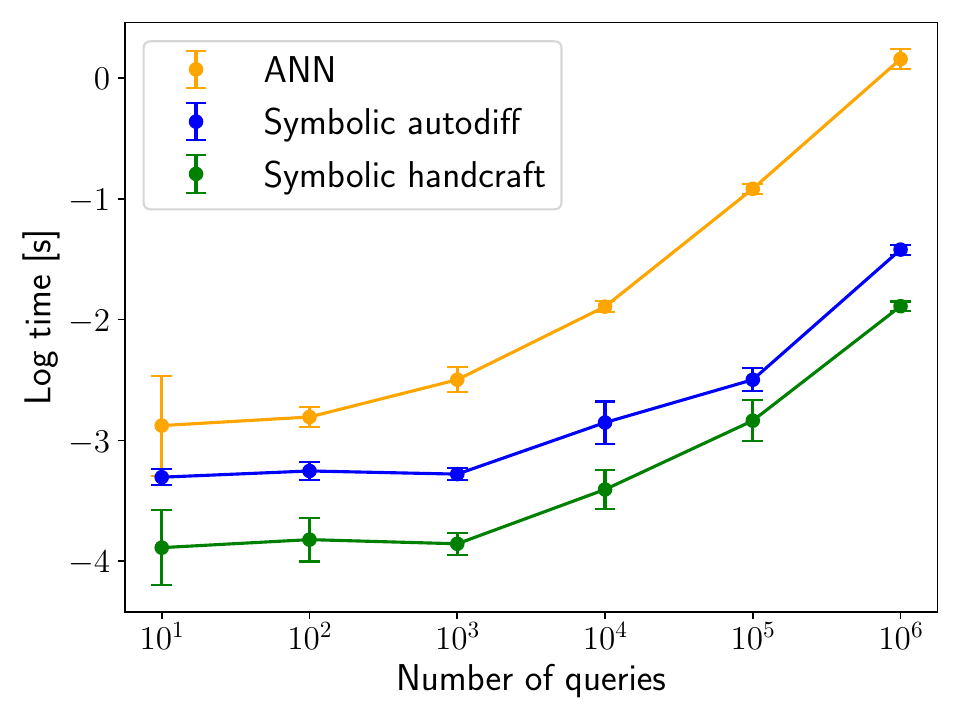}}
  \caption{
  Comparing execution times between an artificial neural network and a symbolic strain energy functional.  Shown are mean and standard deviation trends for (a) the strain energy and (b) its gradient with respect to the invariants. The gradient of a symbolic equation can be hand-crafted or obtained through automatic differentiation, as showcased in (b). The horizontal axis represents the number of query points in each experiment. Each experiment is replicated for 20 times with distinct random seeds.
}
  \label{fig:pr2-time-comp}
\end{figure}

\section{Discussion and Analysis}
\label{sec:discuss}
Thanks to the symbolic forms, classical analysis can be applied to investigate the entire range of the deformation gradient, determining when polyconvexity and growth conditions are violated. In our specific formulation, we can assess quasiconvexity by examining the convexity and non-decreasing conditions of shape functions $\psi_i(I_i)$. To verify coercivity, we analyze the energy functional's growth behavior as the invariants approach infinity, i.e., $I_1, I_2 \to \infty$. It is important to note that our goal is not to find the tightest bounds on the coefficients but rather to ensure that the derived equations satisfy the required conditions.

\subsection{Treloar Machine Learning Model}
The convexity $\psi_2(I_2) = c_{21}I_2$ is evident.  \alert{The non-decreasing condition requires $c_{21}\ge 0$}. For $\psi_1(I_1) = c_{11}I_1 + c_{12}\exp(c_{13}I_2)$ to be convex, it is required that $c_{12} \ge 0$, which can be confirmed by checking the positivity of its second-order derivative.  \alert{This function will be non-decreasing if its first derivative is non-negative, which can be easily determined when $c_{11} \ge 0$.}
By employing the Taylor series expansion for the exponential function and considering any positive integer $p$, we obtain the following inequality,

\begin{equation}
\psi(I_1, I_2) \ge \frac{c_{12}^p}{p!} I_1^p + c_{21}I_2.
\end{equation}

To ensure asymptotic satisfaction of Eq.~\eqref{eq:growth}, we can set $\alpha = \min(c_{12}^p/p!, c_{21})$ with $c_{21} > 0$ and $q = 1$, where $p$ is chosen to be a large value approaching infinity. However, it is essential to note that even for a very large $p$, the condition $q \ge p/(p-1)$ is marginally violated, and the fulfillment of conditions $c_{21} > 0$ and $c_{12} > 0$ may not guarantee coercivity based on the presented justification.

\subsection{Particle-reinforced Composite Machine Learning Model}

The convexity of $\psi_1(I_1) = c_{11}I_1 + c_{12}I_1^4 + \exp(c_{13}I_1)$ is ensured if its second-order gradient is non-negative, i.e., $12c_{12}I_1^2 + c_{13}^2\exp(c_{13}I_1) \ge 0$. Thus, it requires only that $c_{12} \ge 0$ to satisfy convexity.  
\alert{
To evaluate the non-decreasing condition, we consider $\psi_1^\prime(I_1) = c_{11} + 4c_{12}I_1^3 + c_{13} \exp(c_{13}I_1) \geq 0$. First, note that $I_1 = \lambda_1^2 + \lambda_2^2 + \lambda_3^2 > 0$. For very large values of $I_1 \rightarrow +\infty$, the exponential term will dominate, implying $c_{13} \geq 0$. Under the incompressibility condition ($\lambda_3 = 1 / (\lambda_1 \lambda_2)$), it can be shown that $I_1 \geq 2$ serves as a lower bound. Therefore, it suffices to have $c_{11} \geq -(32c_{12} + c_{13} \exp(2c_{13}))$.
}

Similarly, for $\psi_2(I_2) = c_{21}I_2 + c_{22}I_2^2\ln(I_2)$, the convexity requirement is $c_{22}(2\ln(I_2) + 3) \ge 0$, which is satisfied if $c_{22} \ge 0$. 
\alert{
For the non-decreasing condition, considering $I_2 = \lambda_1^2\lambda_2^2 + \lambda_1^2\lambda_3^2 + \lambda_2^2\lambda_3^2 \geq 0$, one can determine that the sufficient condition is $c_{21} \geq 0$.
}

By utilizing the Taylor series expansion for the exponential function, we have,

\begin{equation}
\psi_1(I_1) \ge c_{11}I_1 + c_{12}I_1^4 + \frac{c_{13}^4}{4!}I_1^4 \ge \left(c_{12} + \frac{c_{13}^4}{4!}\right)I_1^4.
\end{equation}
Furthermore, leveraging the fact that $\ln(x) \ge \frac{x-1}{x}$ for positive $x$ (as $x\ln(x)-x+1$ has a positive slope), we can write,
\begin{equation}
\psi_2(I_2) \ge c_{22}I_2^2 + (c_{21}-c_{22})I_2 \ge \hat{c}_{22}I_2^2,
\end{equation}
where $0 \le \hat{c}_{22} \le c_{22}$ accounts for the possibility of $c_{21} < c_{22}$.

Based on these bounds, if we choose $\alpha = \min(c_{12} + c_{13}^4/4!, \hat{c}_{22})$ and set $p=4 \ge 2$ and $q=2 \ge p/(p-1)$, the inequality in Eq.~\eqref{eq:growth} is satisfied, provided that $\alpha > 0$. Consequently, the sufficient conditions to ensure coercivity are $c_{22} > 0$ and $c_{12} > 0$.

\alert{
Notice that, within the training data regime of symbolic regression, and even slightly beyond that, the discovered symbolic function is non-decreasing; see the derivative function in the right Figure \ref{fig:pr2Psi1Symb}, which is positive. 
However, based on the above discussion, the equation found through symbolic regression violates the non-decreasing condition for $I_1 \rightarrow +\infty$ since $c_{13} < 0$. The learned neural network function $\psi_1(I_1)$ does not suffer from this issue because it satisfies this condition by construction. Since such a constraint has not been encoded during the search in symbolic regression, we may violate this condition, particularly beyond the training regime of symbolic regression.
To mitigate such issues at the symbolic regression discovery step, the simplest way would be to generate data from the learned neural network with a wider coverage of $I_1$ and denser data sampling to train the symbolic regression.  A more rigorous approach would be to limit the symbolic regression to search only in the space of all convex and non-decreasing functions. However, incorporating this second restriction is a challenging task for symbolic regression as a discrete search algorithm and requires further research in this area.
}

\section{Conclusion}
\label{sec:concl}
In this work, we present an efficient machine learning framework capable of directly inferring the mathematical expressions of multivariate elastic stored energy functionals from data that fulfill physics principles (e.g. , polyconvexity, material frame indifference) without sacrificing expressivity. A salient feature of this proposed model is the introduction of input convex neural networks for constructing the feature space.  This treatment enables us to retain the convexity of the feature mapping that ensures polyconvexity of the learned neural network hyperelastic models.  By expressing polyconvex hyperelastic models in the feature space, we can then apply SR on each of the feature space basis. This setting enables us to overcome the challenging curse-of-high-dimensionality issue well known in the SR literature while deducing interpretable mathematical models that can be easily comprehended, analyzed, and implemented. 
The compactness of the mathematical models also leads to computationally efficient constitute updates that require less arithmetic operations than those of deep neural networks and hence are well-suited for large-scale physics simulations where repeated executions of the models must be carried out in a large number of integration points. 

\section*{Acknowledgement}
The authors are supported by the National Science Foundation under grant contract CMMI-1846875, and the Dynamic Materials and Interactions Program from the Air Force Office of Scientific Research under grant contracts FA9550-21-1-0391 and FA9550-21-1-0027, 
and the MURI Grant No. FA9550-19-1-0318. These supports are gratefully acknowledged. 
The authors would like to thank Nhon N.  Phan for their valuable feedback and suggestions, which have improved the readability of the manuscript.
The views and conclusions contained in this document are those of the authors and should not be interpreted as representing official policies, either expressed or implied, of the sponsors, including the Army Research Laboratory or the U.S. Government. The U.S. Government is authorized to reproduce and distribute reprints for Government purposes notwithstanding any copyright notation herein.

\section*{Data Availability Statement}
Data supporting the findings of this study will be available at Mendeley Data upon the acceptance of this manuscript.  Codes used to generate the model will be available on GitHub upon the acceptance of this manuscript as well. 

\begin{appendices}

\section{Basic Convexity Relationships}
\label{appx:convx}

\begin{lemma}
A tensor-valued function $W(\tensor{F})$ is convex with respect to $\tensor{F}$ if its Hessian has the following property such that for any $\tensor{F}$ and $\delta \tensor{F}$,
\begin{equation}
\delta{\tensor{F}} : \frac{\partial^2 W(\tensor{F})}{\partial \tensor{F} \partial \tensor{F}} : \delta{\tensor{F}} \ge 0.
\end{equation}
\end{lemma}
\begin{proof}
See Lemma B.3 in \citep{schroder2003invariant}.
\end{proof}

\begin{lemma}\label{lemma:convxI1}
The mapping $W(\tensor{F}): \tensor{F}\to \text{tr}(\tensor{F}^T \tensor{F})$ is convex.  As a result,  the first invariant $I_1(\tensor{C})$ is a convex function of $\tensor{F}$.
\end{lemma}
\begin{proof}
Hessian of the function with respect to $\tensor{F}$ is as follows:
\begin{equation}
\frac{\partial^2 W }{\partial F_{iJ} \partial F_{kL}} = 2 \delta_{ki} \delta_{LJ}.
\end{equation}
By setting $\tensor{A} = \delta \tensor{F}$,  the convexity criterion becomes as follows,
\begin{equation}
\tensor{A} : \frac{\partial^2 W(\tensor{F})}{\partial \tensor{F} \partial \tensor{F}} :\tensor{A}  
= 2 A_{iJ} \delta_{ki} \delta_{LJ} L_{kL}
= 2 A_{iL} A_{iL} \ge 0.
\end{equation}
\end{proof}

\begin{lemma}
The second invariant $I_2(\tensor{C})$ is a convex function with respect to $\text{adj}(\tensor{F})$.
\end{lemma}
\begin{proof}
One can obtain the following relationships for the second invariant,
\begin{align}
I_2
&=
\text{tr}
\left(
	\text{adj}(\tensor{F}^T \tensor{F})
\right)\\
&=
\text{tr}
\left(
	\text{adj}(\tensor{F})
	\text{adj}(\tensor{F}^T)
\right) \quad \text{from Lemma A.6.3 in \citep{schroder2003invariant}}\\
&=
\text{tr}
\left(
	\text{adj}(\tensor{F})
	\text{adj}(\tensor{F})^T
\right) \quad \text{from Lemma A.6.4 in \citep{schroder2003invariant}}\\
&=
\text{tr}
\left(
	\text{adj}(\tensor{F})^T
	\text{adj}(\tensor{F})
\right).
\end{align}
Utilizing the last equation, in conjunction with Lemma \ref{lemma:convxI1}, we deduce that the second invariant is a convex function with respect to $\text{adj}(\tensor{F})$.
\end{proof}

\begin{lemma}
The third invariant $I_3(\tensor{C})$ is a convex function with respect to $\text{det}(\tensor{F})$.
\end{lemma}
\begin{proof}
By utilizing basic properties of matrix operations, we can derive the following relationships:
\begin{align}
I_3=
\text{det}(\tensor{C})
=
\text{det}(\tensor{F}^T\tensor{F})
=
\text{det}(\tensor{F}^T) \text{det}(\tensor{F})
=
\text{det}(\tensor{F})^2.
\end{align}
Therefore, the third invariant is a quadratic function of $\text{det}(\tensor{F})$, making it convex in relation to this determinant.
\end{proof}

\section{Model Setup and Hyperparameters}
In the PNAM, for each shape function, both examples employ ICNN models using a single hidden layer consisting of 50 neurons. The activation functions for the first and second layers are softplus and softplus2, respectively. To ensure non-negativity of the weights, ReLU activation is utilized. The optimization process utilizes the ADAM optimizer with an initial learning rate of 0.001. Strain invariants are normalized using the MinMaxScaler method in both examples. Furthermore, in the first example, the stress data undergoes a pre-processing step and is scaled by a factor of 0.05.

\subsection{Treloar's Data}
\label{sec:paramTrel}
In the vanilla ANN model, one hidden layer with 50 neurons is utilized with an ELU activation layer \citep{clevert2015fast} (see the following remark.) Its training parameters remain consistent with those of the PNAM.

In the SR configuration of $\psi_1(I_1)$, we confine the binary operators to addition and multiplication, while the unary operator is limited to the exponential function (exp). The time budget is capped at 6 minutes or a maximum of 100 iterations, whichever comes first. The expression trees' maximum depth and size are both set to 30. As a result of the observed simplicity in the learned ICNN function for $\psi_2(I_2)$, the maximum size and depth of the expression trees is reduced to 10, and the exponential function is excluded as a possible operator for applying SR to find $\psi_2(I_2)$.

\remark{
The ELU activation function is defined as follows,
\begin{equation}
	\text{ELU}(x) = x H(x) + (\exp(x) - 1) (1 - H(x)),
\end{equation}
where $H(x)$ is the Heaviside function.
}

\subsection{Composite Data}
\label{sec:paramComp}
In the SR configuration of $\psi_1(I_1)$, the binary operators are restricted to addition and multiplication, and the unary operators are limited to the exponential (exp) and logarithm (ln) functions. The time budget is set to 6 minutes or a maximum of 100 iterations, whichever is reached first. The expression trees are constrained to a maximum depth of 10 and a size of 50. Classical energy functionals rarely employ nested ln or exp operators; therefore, during the SR iterations, we exclude the possibility of such combinations.  A similar SR configuration is employed for $\psi_2(I_2)$ with the following differences: the expression trees have a maximum size of 15 and a maximum depth of 5.

\section{Remark on input data normalization}
Input data normalization is a practical technique to enhance the training of machine learning models \citep{goodfellow2016deep}, particularly when the input or output has more than one feature. StandardScaler and MinMaxScaler are two common normalization methods. In StandardScaler, the input features are linearly mapped to have zero mean and unit standard deviation, i.e., $(\vec{x} - \vec{\mu}_x) / \vec{\sigma}_x$ where $\vec{\mu}_x$ and $\vec{\sigma}_x$ are mean and standard deviation per feature, respectively. In MinMaxScaler, the input features are linearly mapped to have zero min and unit max, i.e., $(\vec{x} - \vec{m}_x) / (\vec{M}_x - \vec{m}_x)$ where $\vec{m}_x$ and $\vec{M}_x$ are minimum and maximum values per feature, respectively.  

Since our formulation requires the convexity of the built neural network, the following lemma is provided to ensure that common linear normalization does not affect our utilized formulation

\begin{lemma}\label{lemma:affineConvx}
Affine transformation preserves convexity.
\end{lemma}
\begin{proof}
    Assuming the convexity of $f(\vec{x})$ in $\vec{x} \in \mathbb{R}^n$ and defining $g(\vec{x})= f(\tensor{A}\vec{x} + \vec{b})$ where $\tensor{A}\in \mathbb{R}^{n\times n}, \vec{b}\in\mathbb{R}^n$ are constants, then for two arbitrary $\vec{x}_1, \vec{x}_2 \in \mathbb{R}^n$
    \begin{align}
        g\left(
            \gamma \vec{x}_1 + (1-\gamma) \vec{x}_2
        \right)
        &=
        f\left(
            \vec{A}
            \left(
                \gamma \vec{x}_1 + (1-\gamma) \vec{x}_2
            \right)
            +
            \vec{b}
        \right)
        \\
        &=
        f\left(
            \gamma 
                \left(
                    \vec{A} \vec{x}_1 + \vec{b}
                \right)
            +
            (1-\gamma)
                \left(
                    \vec{A} \vec{x}_2 + \vec{b}
                \right)
        \right)
        \\
        &\le
            \gamma 
                f\left(
                    \vec{A} \vec{x}_1 + \vec{b}
                \right)
            +
            (1-\gamma)
                f\left(
                    \vec{A} \vec{x}_2 + \vec{b}
                \right)
            \quad \text{using convexity of} \ f
        \\
        &=
        \gamma g(\vec{x}_1)
        +
        (1-\gamma) g(\vec{x}_2).
    \end{align}
    
\end{proof}

The affine transformation lemma guarantees that the usual linear data normalization (standardization) in machine learning, such as \texttt{StandardScaler} and \texttt{MinMaxScaler} \cite{pedregosa2011scikit}, do not alter the convexity results.

\section{Finite Element Discretization}
For ease of implementation, we use penalty approach to weakly impose incompressibility. Although this method does dot enforce incompressibility exactly, it leads to a simple system of nonlinear equations based solely on the displacement field. Other multi-field formulations \citep{weiss1996finite,schroder2017stability} are possible but they are computationally more demanding which is not considered in this work.

For numerical stability reasons, one may prefer to additively decompose the strain energy functional to volumetric $\psi_{\text{vol}}$ and deviatoric $\psi_{\text{dev}}$ parts, i.e., $\psi(\tensor{C}) = \psi_{\text{dev}}(\bar{\tensor{C}}) + \psi_{\text{vol}}(J)$. In the nearly incompressible regime, the volumetric part acts as a penalty term to enforce $|J - 1| < \epsilon$. In this work, the penalty term is chosen as follows,
\begin{equation}
    \psi_{\text{vol}} = 
    \frac{1}{2} \kappa \ln^2(J),
\end{equation}
where $\kappa$ is the penalty parameter. Higher value of this parameter enforces incompressibility more strictly, however for numerical instability issues one may need to choose it reasonably high. The deviatoric part of the right Cauchy-Green is defined,
\begin{equation}
    \bar{\tensor{C}} = 
    J^{-\frac{2}{3}} \tensor{C}.
\end{equation}

The unknown displacement field $\vec{U}(\vec{X})$ defined over the undeformed domain $\Omega_0$ is the stationary point of the total energy functional when neglecting body forces and having non-zero traction boundary conditions, 
\begin{align}
    &\vec{U} = \underset{\vec{U}}{\argmin} \ \Pi(\vec{U}(\vec{X})),
    \\
    &\Pi(\vec{U}) = \int_{\Omega_0} \psi(\vec{U}(\vec{X})) \ d\Omega.
\end{align}

To find the minima, we discretize the domain using hexahedral elements and use linear basis functions for the displacement field,
\begin{align}
    \Pi(\vec{U}) \approx 
    \Pi(\bar{\vec{U}}) =
    \sum_{i=1}^{N_{\text{elem}}}
    \int_{\Omega^{e}_i} 
    \psi\left(
    \vec{N}^T(\vec{X})\bar{\vec{U}}
    \right) \ d\Omega,
\end{align}
where interpolation using finite element basis is utilized, i.e., $\vec{U} = \vec{N}^T(\vec{X})\bar{\vec{U}}$, and $\vec{N}$ is the linear basis function for the vector field in 3D. The element-level integration is performed using a second-order quadrature rule. To solve the nonlinear system of equations based on the nodal displacement $\bar{\vec{U}}$, we employ the Newton-Raphson method. For more detailed information, readers are referred to \cite{belytschko2014nonlinear}. 
Our implementation is carried out using \texttt{FEniCS} \citep{alnaes2015fenics}.

\end{appendices}

\bibliographystyle{plainnat}
\bibliography{main}

\end{document}